\documentclass[sigconf]{acmart}

\usepackage{inputenc} 
\usepackage[T1]{fontenc}    
\usepackage{hyperref}       
\usepackage{url}            
\usepackage{booktabs}       
\usepackage{amsfonts}       
\usepackage{nicefrac}       
\usepackage{microtype}      
\usepackage{float}
\usepackage{algorithm}
\usepackage{makecell}
\usepackage{multirow}
\usepackage{hhline}
\usepackage{algorithmic}
\usepackage{amsmath,bm}
\usepackage{amssymb}
\usepackage[shortlabels]{enumitem}
\usepackage{mathrsfs}
\usepackage{graphicx}
\usepackage{color}
\usepackage{soul}
\usepackage{wrapfig}
\usepackage{subfigure}

\newcommand{\Ac}{\mathcal{C}}
\renewcommand{\a}{\mathbf{a}}

\newcommand{\Bc}{\mathcal{B}}
\renewcommand{\b}{\mathbf{b}}
\newcommand{\bxx}{\bar{\mathbf{x}}}
\newcommand{\bx}{\bar{x}}

\newcommand{\D}{\mathcal{D}}
\newcommand{\Dc}{\mathscr{Dc}}

\renewcommand{\d}{\mathbf{d}}

\newcommand{\Eb}{\mathbb{E}}

\newcommand{\ee}{\acute{\text{e}}}

\newcommand{\f}{\mathbf{f}}
\newcommand{\tf}{\tilde{f}}

\newcommand{\Fc}{\mathcal{F}}

\newcommand{\g}{\mathbf{g}}

\newcommand{\I}{\mathbf{I}}

\newcommand{\Mcc}{\mathcal{M}}

\newcommand{\Q}{\mathbf{Q}}
\newcommand{\q}{\mathbf{q}}
\newcommand{\R}{\mathbf{R}}

\newcommand{\W}{\mathbf{W}}

\newcommand{\x}{\mathbf{x}}

\newcommand{\y}{\mathbf{y}}
\newcommand{\z}{\mathbf{z}}
\newcommand{\xb}{\mathbf{\bar{x}}}

\newcommand{\zt}{\mbox{\boldmath{$\zeta$}}}

\newcommand{\0}{\mathbf{0}}
\newcommand{\1}{\mathbf{1}}

\renewcommand{\L}{V}

\newcommand{\bepi}{\bm{\epsilon}}

\newtheorem{thm}{Theorem}
\newtheorem{cor}[thm]{Corollary}
\newtheorem{lem}{Lemma}

\newtheorem{prop}[thm]{Proposition}

\newtheorem{defn}{Definition}

\newtheorem{rem}{Remark}

\newtheorem{assum}{Assumption}

\newcommand{\veta}{\boldsymbol \eta}

\newtheorem{lemma}{Lemma}

\AtBeginDocument{%
  \providecommand\BibTeX{{%
    \normalfont B\kern-0.5em{\scshape i\kern-0.25em b}\kern-0.8em\TeX}}}

\setcopyright{acmcopyright}
\copyrightyear{2018}
\acmYear{2018}
\acmDOI{10.1145/1122445.1122456}

\acmConference[MobiHoc '20]{MobiHoc '20: ACM International Symposium on Theory, Algorithmic Foundations, and Protocol Design for Mobile Networks and Mobile Computing}{June 30 -- July 03, 2020}{Shanghai, China}
\acmBooktitle{MobiHoc '20: ACM Symposium on Theory, Algorithmic Foundations, and Protocol Design for Mobile Networks and Mobile Computing,
  June 30 -- July 03, 2020, Shanghai, China}
\acmPrice{15.00}
\acmISBN{978-1-4503-XXXX-X/18/06}

\begin{document}

\title[Private and Communication-Efficient Learning over Networks]{Private and Communication-Efficient Edge Learning: A Sparse Differential Gaussian-Masking Distributed SGD Approach}


%
%
%
  
\author{Xin Zhang$^{*+}$ \quad Minghong Fang$^*$ \quad Jia Liu$^*$ \quad Zhengyuan Zhu$^+$}
\affiliation{
	\institution{$^*$Department of Computer Science, Iowa State University}
	\institution{$^+$Department of Statistics, Iowa State University}
	\city{Ames, IA 50011, U.S.A.}
}

\renewcommand{\shortauthors}{Zhang, et al.}

\begin{abstract}
With rise of machine learning (ML) and the proliferation of smart mobile devices, recent years have witnessed a surge of interest in performing ML in wireless edge networks.
In this paper, we consider the problem of jointly improving data privacy and communication efficiency of distributed edge learning, both of which are critical performance metrics in wireless edge network computing.
Toward this end, we propose a new decentralized stochastic gradient method with sparse differential Gaussian-masked stochastic gradients (SDM-DSGD) for non-convex distributed edge learning.
Our main contributions are three-fold:
i) We propose a generalized differential-coded DSGD update, which enable a much lower transmit probability for gradient sparsification, and provide an $\tilde{O}(1/\sqrt{NT})$ convergence rate;
ii) We theoretically establish the privacy and communication efficiency performance guarantee for our SDM-DSGD method, which outperforms all existing works;
and iii) We reveal theoretical insights and offer practical design guidelines for the interactions between privacy preservation and communication efficiency, two conflicting performance goals.
We conduct extensive experiments with a variety of learning models on MNIST and CIFAR-10 datasets to verify our theoretical findings.
Collectively, our results contribute to the theory and algorithm design for distributed edge learning.

\end{abstract}

\ccsdesc[500]{Computing methodologies~Distributed algorithms}
\ccsdesc[300]{Computing methodologies~Machine learning}
\ccsdesc[300]{Security and privacy}
\ccsdesc[300]{Networks~Network performance analysis}

\keywords{Edge network computing, distributed learning optimization, communication efficiency, differential privacy.}

\maketitle

\section{Introduction}\label{Section: introduction}

In recent years, advances in machine learning (ML) have enabled many new and emerging applications that transform human lives. 
Traditionally, the training of of ML applications often rely on  cloud-based data-centers to collect and process vast amount of data. 
With the proliferation of smart mobile devices and IoT (Internet-of-Things), data and requests for ML are increasingly being generated by devices from wireless edge networks.
Due to high latency, low bandwidth, and privacy concerns\footnote{Data privacy concerns in designing decentralized learning algorithms have long been raised in the literature (see, e.g., \cite{jayaraman2018distributed,cheng2018leasgd,zhang2018admm}).
This is because in many applications (e.g., healthcare, finance, recommendation systems, etc.), ML models are usually trained by data over distributed systems that contain users' privacy information.}, 
collecting all data to the cloud for processing may no longer be feasible or desirable. 
Therefore, the hope of ``ML at the wireless edge'' (``edge ML'' for short) is to retain data in wireless edge networks and perform ML training {\em distributively} across end-user devices and edge servers (or called edge clouds).
By doing so, one could potentially improve edge ML training performance while ensuring user privacy.

However, the successful deployment of edge ML faces significant technical hurdles.
During the execution of distributed edge ML algorithms, each node in the network needs to exchange information with its local neighbors, which often injects intensive communication load into the network.
This problem is further exacerbated by the inherent capacity constraints of wireless channels (due to channel fading, interference, etc.) and edge devices (due to limits in transmitter power, receiver sensitivity, etc.). 
Moreover, merely keeping data at edge devices does {\em not} ensure privacy: the released local messages that are exchanged over the air during each iteration of the algorithm still allow adversaries to infer the local sensitive data~\cite{abadi2016deep,zhang2018admm,jayaraman2018distributed,cheng2018leasgd}.
Indeed, existing ML algorithms have not been designed for the edge environments.
Hence, there is a pressing need for a fundamental understanding on how to design distributed algorithms to ensure {\em both} communication-efficient and privacy-preserving edge ML under severe communication and privacy constraints.


Unfortunately, the design of private and communication-efficient edge ML training faces two inherently {\em conflicting} challenges:
i) On one hand, to preserve users' privacy while ensuring training convergence, one often needs to inject certain {\em unbiased i.i.d.} (independent and identically distributed) random noise into the exchanged information.
However, the privacy guarantee with i.i.d. noise typically degrades with the number of training iterations $T$~\cite{abadi2016deep,jayaraman2018distributed,wang2019dp,zhang2017efficient}.
Hence, for a given privacy constraint, there exists a maximum value of $T$ that a distributed edge ML training algorithm can run;
ii) On the other hand, to design communication-efficient distributed edge ML algorithms, it is necessary to perform gradient sparsification and/or compression (see, e.g., \cite{tang2018communication,stich2018sparsified,wangni2018gradient,reisizadeh2019robust}).
However, as we show later, these operations induce a {\em different} type of random noise, which increases the the value of $T$ for achieving some desired training loss while providing {\em no} privacy guarantee.
It is highly challenging to reconcile these two conflicting types of randomness in distributed edge ML algorithmic design.

To date, results on edge ML algorithmic designs that are both private and communication-efficient remain rather limited.
Most of the existing work focus on either communication efficiency \cite{tang2018communication,zhang2019compressed,reisizadeh2019exact} or data privacy \cite{liu2018differentially,ding2018consensus}.
For the limited amount of work that considered both, they are either only restricted to the server/worker architecture or having unsatisfactory performances and high implementation complexity (see Section~\ref{sec:related} for more in-depth discussion).
The above limitations of the existing work motivate us to propose a new decentralized stochastic gradient descent (DSGD) method with \underline{s}parse \underline{d}ifferential Gaussian-\underline{m}asked stochastic gradients.
For convenience, we refer to this method as SDM-DSGD.
Our SDM-DSGD method addresses the aforementioned technical challenges and offer significantly improved privacy and communication efficiency performances.
Our main results and their significance are summarized as follows:

\begin{list}{\labelitemi}{\leftmargin=1em \itemindent=0.em \itemsep=.2em}

\item 
We propose a SDM-DSGD method for non-convex distributed edge ML training, which is differentially-private, communication-efficient, and applicable for general network topologies.
We show that, with the properly chosen parameters, our SDM-DSGD algorithm is $(\epsilon,\delta)$-differentially private (DP) and enjoys an $\tilde{O}(1/\sqrt{NT})$ convergence rate, where $N$ is the number of nodes in the network and $T$ is the final iteration index of the algorithm.
Moreover, we show that the maximum value of $T$ scales as $O(m^{4})$, where $m$ is the size of local dataset at each node\footnote{Our algorithms can straightfowardly be extended to cases with datasets having unbalanced sizes, i.e., $m_{n_1} \ne m_{n_2}$ for $n_{1} \ne n_{2}$.}.

\item It is also worth pointing out that our SDM-DSGD is a {\em generalized} differential-coded DSGD approach in the sense that: {\em All} existing differential-coded DSGD algorithms can be seen as a special case of our SDM-DSGD (e.g., \cite{tang2018communication}).
Specifically, our key updating step in SDM-DSGD is a linear combination of the current state and the standard DSGD update.
Remarkably, with this generalized updating step, one can perform gradient sparsification with a much lower transmit probability $p$, which implies significantly improved privacy.
This also relaxes the restricted constraint in \cite{tang2018communication} on finding a ``valid'' $p$.
We also note that, thanks to this new algorithmic structure, the non-private version of our algorithm (i.e., no Gaussian-masking) may be of independent interest.

\item Based on our theoretical results from SDM-DSGD, we go one-step further to investigate the interactions between i) gradient differential sparsification and ii) Gaussian masking, which are the two key components responsible for communication efficiency and privacy in SDM-DSGD, respectively.
Toward this end, we compare an alternative design that also has the same two components but with their order being reversed.
Our analysis shows that the proposed SDM-DSGD scheme is superior and can reduce the privacy budget by a $p^2$-fraction.
This insight deepens our understanding on these two components and offers algorithmic design guidelines in practice.

\item Lastly, we conduct extensive experiments to examine the performance of our SDM-DSGD algorithm with a variety of deep learning models on MNIST and CIFAR-10 datasets. 
Our experiments show that the accuracy of SDM-DSGD outperforms two state-of the-art decentralized learning algorithms~\cite{lian2017can,tang2018communication} under the same communication cost and privacy budget.
These experiments corroborate our theoretical results.

\end{list}

Collectively, our results in this paper contribute to the state of the art of theories and algorithm design for communication-efficient and privacy-preserving decentralized learning.
The rest of the paper is organized as follows.
In Section~\ref{Section: preliminary}, we will review necessary background for our algorithm design and analysis.
In Section~\ref{sec:sdm_dsgd}, we introduce our SDM-DSGD algorithm and then analyze its performances in privacy and convergence.
Numerical results are provided in Section \ref{Section: experiment}. 
In Section \ref{Section: conclusion}, we provide concluding remarks.

\section{Related Work} \label{sec:related}

As mentioned in Section \ref{Section: introduction}, results on private and communication-efficient distributed learning algorithms remain quite limited in the literature.
For example, to achieve both communication efficiency and differential privacy, Agarwal {\em et al.}~\cite{agarwal2018cpsgd} proposed the cpSGD algorithm based on the randomized quantization and Binomial masking.
It is shown both theoretically and experimentally that Binomial masking achieves nearly the same utility as Gaussian masking, while the communication cost is significantly reduced.
However, this work mainly focused on the distributed mean estimation (DME) problem under the server/worker architecture. 
It remains unclear how to implement the cpSGD algorithm to train general deep learning models in networks with general communication topologies.
Another related work is \cite{cheng2018leasgd}, where Cheng {\em et al.} proposed a new decentralized algorithm named leader-follower elastic averaging stochastic gradient descent (LEASGD).
In the LEASGD algorithm, the computation nodes are dynamically categorized into two pools: leader pool with nodes of lower loss values and follower pool with nodes of higher loss. 
In each iteration, the leader nodes will pair with followers to guide the followers in the right direction.
Gaussian masking was adopted in the communication step to protect the data privacy.
Although this work numerically showed LEASGD's communication efficiency, the implementation complexity of LEASGD is high due to  the categorization and lead-follower pairing in each iteration.  
Also, the theoretical performance of LEASGD is unclear under the non-convex cases. 
In contrast to these existing work, in this paper we propose a communication-efficient and privacy-preserving distributed training algorithm named SDM-DSGD, for distributed nonconvex learning. 
Our SDM-DSGD algorithm can be viewed as a variant of state-of-the-art decentralized learning algorithm DSGD~\cite{lian2017can} by introducing the randomized sparification and the Gaussian mechanism.

\section{Differential Privacy and Gradient Sparsification: A Primer} \label{Section: preliminary}
To facilitate subsequent technical discussions on privacy and communication efficiency, in this section, we provide the necessary background on differential privacy and gradient sparsification.

\smallskip
\textbf{1) Differential Privacy:}
Differential privacy (DP)~\cite{dwork2014algorithmic,dwork2006calibrating} is a canonical privacy metric for the privacy-preserving data analysis.
Under the DP framework, privacy is defined and measured by how noticeable the distribution of the outcome of some query mechanism changes when only one sample in the dataset is changed:
\begin{defn}[$(\epsilon,\delta)$-Differential Privacy~\cite{dwork2006calibrating}] \label{Def: DP}
Two datasets $\D$ and $\D^\prime$ are said to be adjacent if and only if they differ by only one element.
Given two adjacent $N$-element datasets $\D$ and $\D^\prime \in \Dc^N$ and $\epsilon,\delta >0$, a randomized query mechanism $\Mcc:\Dc^N\rightarrow \mathbb{R}^d$ is called $(\epsilon,\delta)$-differentially-private ($(\epsilon,\delta)$-DP) if and only if for any measurable set $E \in \mathbb{R}^d,$ the output of $\Mcc$ satisfies $\mathbb{P}(\Mcc(\D)\in E) \le e^{\epsilon} \mathbb{P}(\Mcc(\D^\prime)\in E) +\delta$.
\end{defn}

In the literature, two popular approaches to achieve DP are the so-called Gaussian and Laplacian masking mechanisms, both of which share the same form: $\Mcc(\D) = q(\D) + \eta,$ where $q(\cdot)$ is a query function and $\eta$ is the injected Gaussian or Laplacian masking noise.
In our work, we focus on the Gaussian masking mechanism.

\smallskip
\textbf{2) Sparsification:}
In the literature, sparsification, also known as sparse compression, is a commonly used compression technique for compressing gradients to design communication-efficient distributed learning algorithms \cite{wangni2018gradient,tang2018communication,stich2018sparsified}.
The key idea of sparsification is to apply the the following Bernoulli randomized operation to sparsify a high-dimensional vector:
\begin{defn}[Sparsifier~\cite{wangni2018gradient}] \label{Def: Sparse}
For any vector ${\x}=[x_1,\cdots,x_d]^{\top} \in \mathbb{R}^d$ and a constant $p\in [0,1)$, $S({\x})$ outputs a sparse vector with the $i$-th element $[S{(\x)}]_{i}$ following the Bernoulli$(p)$ distribution:
\begin{align*}
\begin{cases}
\mathrm{Pr} ([S{(\x)}]_{i} = \frac{\x_i}{p} ) = p,\\
\mathrm{Pr} ([S{(\x)}]_{i} = 0 ) = 1-p.
\end{cases}
\end{align*}
\end{defn}
We can see that the sparsifier operation randomly selects some coordinates and sets the information in these coordinates to zero.
Also, it follows immediately from Definition~\ref{Def: Sparse} that the sparsification operation is {\em unbiased} and the introduced variance depends on the magnitude of input vector (we omit the proof due to its simplicity):
\begin{lem}
For any $\x\in\mathbb{R}^d,$ the random output $S(\x)$ satisfies: 1) unbiased expectation: $\mathbb{E}\big(S(\x)\big)=\x$; and 2) input-dependent variance: $\text{Var}\big(S(\x)\big)=(1/p-1)\|\x\|_2.$
\end{lem}

It is worth noting that, although sparsification is a random mechanism, it does {\em not} provide any privacy guarantee in terms of DP: 
Given a dataset $\D$ and a query function $q(\cdot),$ there exists an adjacent dataset $\D^\prime$ such that $q(\D) \neq q(\D^\prime).$
Under the sparsification with probability $p,$ consider the event $E = \{q(\D)/p\}$.
Then, for dataset $\D$, to have the output $q(\D)/p,$ it requires that all the coordinates need to be selected. 
Thus, it holds that $\mathrm{Pr}(S(q(\D)\in E) = p^d > 0,$ where $d$ is the dimension.
However, for dataset $\D^\prime,$ because $q(\D) \neq q(\D^\prime),$ it is impossible to have an output as $q(\D)/p,$ which implies $\mathrm{Pr}(S(q(\D^\prime)\in E) = 0.$ 
Thus, it is impossible to find valid $\epsilon$ and $\delta \in [0,1)$ to satisfy the DP inequality in Definition \ref{Def: DP}.
Interestingly, although sparsification does not offer DP, we will show it later that by sparsifying part of the original information, the sparsifier operation does help to improve the privacy protection performance.
 
\section{A Sparse Differential Gaussian- Masking SGD Approach} \label{sec:sdm_dsgd}

In this section, we first present the problem formulation of edge ML training in Section~\ref{subsec:formulation}.
Then, we will present our SDG-DSGD algorithm in Section~\ref{subsec:alg} and its main theoretical results in Section~\ref{subsec:results}.

\subsection{Problem Formulation of Edge ML Training} \label{subsec:formulation}
In this paper, we use an undirected graph $\mathcal{G} = (\mathcal{N},\mathcal{L})$ to represent a wireless edge network with {\em general} network topology, where $\mathcal{N}$ and $\mathcal{L}$ are the sets of nodes and links, respectively, with number of nodes as $|\mathcal{N}| = n$.
We let $\x \in \mathbb{R}^{d}$ denote a global decision vector to be learned or estimated.
In edge ML training, we want to distributively solve an unconstrained optimization problem:
$\min_{\x \in \mathbb{R}^{d}} f(\x;\D)$, and $f(\x;\D)$ can be decomposed node-wise as follows\footnote{Here we assume the datasets are balanced}:
\begin{align} \label{eqn_general_problem}
\min_{\x \in \mathbb{R}^{d}} f(\x;\D) = \min_{\x \in \mathbb{R}^{d}} \sum_{i=1}^{n} f(\x;\D_i),
\end{align}
where $f(\x;\D) = \frac{1}{|\D|}\sum_{\z\in\D} f(\x;\z)$ for any dataset $\D.$
Here, each local objective function $f(\x;\D_i)$ is only observable to node $i$.
It is easy to see that Problem~(\ref{eqn_general_problem}) can be  equivalently reformulated as the following {\em consensus form}:
\vspace{-.05in}
\begin{align} \label{Eq:problem1}
\min \sum_{i=1}^{n} f(\x_i;\D_i),~  s.t. \x_i = \x_j,~ \forall (i,j) \in \mathcal{L}. 
\end{align}
where $\x_i\in \mathbb{R}^d$ is the local copy of $\x$ at node $i$. 
The constraints in Problem (\ref{Eq:problem1}) guarantee that the all local copies are equal to each other, hence the name consensus form.

\subsection{The SDM-DSGD Algorithm}\label{subsec:alg}

As the name suggests, our proposed SDM-DSGD method is inspired by the classical decentralized gradient descent (DGD) algorithm \cite{nedic2009distributed,yuan2016convergence,lian2017can},
which is one of the most effective approaches for distributively solving network consensus optimization problems.
The DGD framework is built upon the notion of {\em consensus matrix}, which is denoted as $\W \in \mathbb{R}^{n \times n}$ in this paper.
Specifically, in each iteration of DGD, each node in the network performs an update that integrates a local (stochastic) gradient step and a weighted average from its neighbors' parameters based on $\W$.
Mathematically, $\W$ satisfies the following properties:
\begin{list}{\labelitemi}{\leftmargin=2em \itemindent=-0.0em \itemsep=.2em}
\item[1)] {\em Doubly Stochastic:} $\sum_{i=1}^{n} [\mathbf{W}]_{ij}=\sum_{j=1}^{N} [\mathbf{W}]_{ij}=1$.
\item[2)] {\em Symmetric:} $[\mathbf{W}]_{ij} = [\W]_{ji}$, $\forall i,j \in \mathcal{N}$. 
\item[3)] {\em Network-Defined Sparsity Pattern:} $[\W]_{ij} > 0$ if $(i,j)\in \mathcal{L}$ and $[\mathbf{W}]_{ij}=0$ otherwise, $\forall i,j \in \mathcal{N}$.
\end{list}
Properties 1)--3) imply that the spectrum of $\W$ (i.e., the set of all eigenvalues) lies in the interval $(-1,1]$ with exactly one eigenvalue being equal to 1.
Also, all eigenvalues being real implies that they can be sorted as $-1 < \lambda_n(\mathbf{W}) \leq \cdots \leq \lambda_1(\mathbf{W}) = 1$.
To facilitate later discussions, we let $\beta \triangleq \max\{|\lambda_2(\mathbf{W})|,|\lambda_n(\mathbf{W})|\} \in (0,1)$, i.e., the second-largest eigenvalue of $\W$ in magnitude.
The use of the consensus matrix is due to the fact that $(\W \otimes \I_{P}) \x = \x$ {\em if and only} if $\x_i = \x_j$, $(i,j) \in \mathcal{L}$\cite{nedic2009distributed}, where $\x = [\x_{1}^{\top},\ldots,\x_{n}^{\top}]^{\top}$ and $\otimes$ represents the Kronecker product.
Therefore, Problem (\ref{Eq:problem1}) can be reformulated as $\min_{\x \in \mathbb{R}^{D}} \sum_{i=1}^{N} f_{i}(\x_{i})$, $\mathrm{s.t.} \,\, (\W \otimes \I_{P}) \x = \x$, which further leads to the original DGD algorithmic design\cite{nedic2009distributed}.
With the notion of $\W$, our SDM-DSGD algorithm can be stated as follows:

\medskip
\hrule 
\vspace{.02in}
\noindent {\textbf{Algorithm~1:} A \underline{S}parse \underline{D}ifferential Gaussian-\underline{M}asking \underline{D}istributed \underline{S}tochastic \underline{G}radient \underline{D}escent (SDM-DSGD) Algorithm.} \label{alg:sdm-dsgd}
\vspace{.02in}
\hrule
\vspace{0.02in}
\noindent {\textbf{ Initialization:}}
\begin{enumerate} [topsep=1pt, itemsep=-.1ex, leftmargin=.2in]
\item[1.] Set the initial state $\x_{i,0}\!=\! \y_{i,0}\!=\! \d_{i,0}\!=\!\0$, $\forall i,$ and $t = 1.$
\end{enumerate}
\noindent {\textbf{ Main Loop:}}
\begin{enumerate} [topsep=1pt, itemsep=-.1ex, leftmargin=.2in]
\item[2.] In the $t$-th iteration, each node sends the sparsified differential $S(\d_{i,t})$ to its neighbors, where $S(\cdot)$ is the sparsifier operation. 
Also, upon collecting all neighbors' information, each node $i \in \mathcal{N}$ updates the following local values:
\begin{enumerate} [topsep=1pt, itemsep=-.1ex, leftmargin=.2in]
\item[a) ] Reconstruct node $i$'s neighbors inexact copies:
$\x_{j,t}\!=\!\x_{j,t-1}\!+\! S(\d_{j,t-1}),$  $\forall j \in \mathcal{N}_i;$
\item[b) ] Update local copy: 
$\y_{i,t}=(1-\theta)\x_{i,t} + \theta \big(\sum_{j \in \mathcal{N}_{i}} [\W]_{ij}\x_{j,t} -  \gamma  \big(\nabla f (\x_{i,t};\zeta_{i,t}) + \eta_{i,t}\big)\big),$ where $\nabla f(\cdot;\zeta_{i,t})$ is the stochastic gradient, and $\eta_{i,t} \sim N(0,\sigma^2\I_d)$ is a Gaussian random noise.
\item[c) ] Compute the local differential: 
$\d_{i,t}\!=\!\y_{i,t} \!-\! \x_{i,t}$.
 \end{enumerate}
\item[3.] Stop if some convergence criterion is met; otherwise, let $t \leftarrow t+1$ and go to Step 2. 
\end{enumerate}
\hrule
\medskip

\begin{rem}{\em
Algorithm~1 is motivated by and bears some similarity with the DGD-type communication-efficient distributed learning in the literature~\cite{tang2018communication,zhang2019compressed}.
In these existing work, rather than exchanging the states directly, the compressed differentials between two successive iterations of the variables are communicated to reduce the communication load.
By contrast, Algorithm~1 differs from these existing work in the following key aspects:
i) As noted in Section~\ref{Section: introduction}, the update in Step~2.b) in SDM-DSGD generalizes the existing work by using a {\em linear combination} of the current state and the DSGD update.
It was shown that when using the sparsification in the proposed algorithm in~\cite{tang2018communication}, the transmit probability $p$ in Definition~\ref{Def: Sparse} is required to be greater than $4(1-\lambda_n)^2/(4(1-\lambda_n)^2 + (1-|\lambda_n|)^2),$ where $\lambda_n$ is the smallest eigenvalue of $\W$.
In contrast, our generalized framework allows a much smaller $p$ in the sparsification, i.e., significantly better communication-efficiency.
ii) In addition to performance gains in terms of communication-efficiency, as will be shown later, our generalized algorithm also improves the convergence speed from $O(T^{-1/3})$ to $\tilde{O}(T^{-1/2})$ as long as $p = \Omega(1/\log(T))$.

}
\end{rem}

Before we state our main theoretical results, it is insightful to offer some intuitions on how our SDM-DSGD method is derived.
Toward this end, we rewrite the update rule of in SDM-DSGD algorithm in the following vector form:
\begin{align}\label{Eq: updating equation}
\left\{
\begin{aligned}
 \x_{t} &= \x_{t-1} + S(\d_{t-1}), \\
 \y_{t} &= (1-\theta)\x_{t} +\theta \Big(\tilde{\W} \x_{t} - \gamma \big(\nabla \f(\x_{t};\zeta_{t}) +\veta_{t} \big)\Big), \\
 \d_{t} &= \y_{t} - \x_{t}, \\
\end{aligned}
\right.
\end{align}
where $\tilde{\W} \triangleq \W \otimes \I_n$ and $\f(\x_{t};\zeta_{t}) = \sum_{i=1}^{n}f(\x_i;\zeta_{i,t})$.
Define a Lyapunov function $\L_\gamma(\x;\D) \triangleq \frac{1}{2}\x^\top(\I - \tilde{\W})\x + \sum_{i=1}^{n}f(\x_i;\D_i),$ and its stochastic version ${\L}_\gamma(\x;\zeta) \triangleq \frac{1}{2}\x^\top(\I - \tilde{\W})\x + \sum_{i=1}^{n}f(\x_i;\zeta_i),$ where $\x = [\x_1^\top,\cdots,\x_n^\top]^\top \in \mathbb{R}^{nd}$ and the random sample batch $\zeta = \{\zeta_i\}_{i=1}^n.$
It can be readily verified that: 
\begin{align}
\left\{
\begin{aligned}
\y_t 
&= \x_t - \theta(\nabla{\L}_\gamma(\x_t;\zeta_t)+\gamma \veta_t), \\
\d_t & = - \theta(\nabla{\L}_\gamma(\x_t;\zeta_t)+\gamma \veta_t),
\end{aligned}
\right.
\end{align}
which implies that the iterates update can be written as: 
\begin{align}\label{Eq: iterate update}
\x_{t+1}\! =\! \x_t \!+\! S(\d_t)  \stackrel{}{\!=\!} \x_t \!-\! \theta\nabla{\L}_\gamma(\x_t;\zeta_t) \!-\! \theta\gamma \veta_t \!+\! \bepi_t,
\end{align}
where $\bepi_t$ represents the noise from the sparsifier.
Therefore, the iterative update rule in SDM-DSGD can be viewed as applying the stochastic gradient descent algorithm on the Lyapunov function $\L_\gamma(\x;\D)$ with two additional noises $\theta\gamma \veta_t$ and $\bepi_t$, one from the privacy protection and the other one from the sparse compression.

\subsection{Main Theoretical Results}\label{subsec:results}

In this subsection, we will establish the privacy and convergence properties of the proposed SDM-DSGD method.
For better readability, we state the main theorems and their key insights in this subsection and relegate the proofs of the main theorems to the appendices.
We start with stating the following assumptions:
\begin{assum}\label{Assumption: function}
The global objective function $f(\cdot)$ satisfies:
\begin{enumerate}[topsep=1pt, itemsep=-.1ex, leftmargin=.3in]
	\item[(1)] Given dataset $\D=\{\D_i\}_{i=1}^{n},$ $f(\x;\D)$ is bounded from below, i.e., $\exists \x^*_{\D}\in\mathbb{R}^d,$ such that $f(\x;\D)\ge f(\x^*_{\D};\D)$, $\forall \x\in\mathbb{R}^d;$
	\item[(2)] The function $f(\x;z)$ is continuously differentiable and has $L$-Lipschitz continuous gradient, i.e., there exists a constant $L >0$ such that $|\nabla f(\x;z) -\nabla f(\y;z)|\le L \| \x-\y \|_2,$ $\forall \x,\y\in\mathbb{R}^d;$
	\item[(3)] The stochastic gradient is unbiased and has bounded variance with respective to the local dataset, i.e. $\Eb_{z\sim\D_i}[\nabla f(\x;z)] = f(\x;\D_i)$ and $\text{Var}_{z\sim \D_i}[\nabla f(\x;z)]\le \tilde{\sigma}^2;$
	\item[(4)] The function $f(\x;z)$ is coordinate-wise $G/\sqrt{d}$-smooth, i.e., for all coordinates $k,$ $|[\nabla f(\x;z)]_k|\le G/\sqrt{d}.$ 
\end{enumerate}
\end{assum}

The first three assumptions are standard for the convergence analysis of stochastic algorithms \cite{ghadimi2013stochastic,lian2015asynchronous,lian2017can}.
The last assumption characterizes the sensitivity of objective function with respect to $\x$ in each coordinate.
Note that it also implies the $\ell_2$-sensitivity bound $\|\nabla f(\x;z)\|\le G,$ which is useful for differential privacy (see Defintion~\ref{Def: l2 sensitivity} in Appendix~\ref{Theorem: Privacy Guarantee} and \cite{wang2019dp,jayaraman2018distributed}).

\smallskip
\textbf{1) Privacy Analysis:}
In our SDM-DSGD algorithm, in the $t$-th iteration, each node releases the local information $S(\d_t)$, which is generated by first applying Gaussian masking on the local stochastic gradient $g(x_{i,t};\zeta_{i,t})$ and then the sparsifier operation on the local differential $d_t$. 
Consider the latent vectors $\omega_{i,t} \sim \text{Bin}(d,p)$ and $\omega_t = [\omega_{1,t}^\top,\cdots, \omega_{n,t}^\top]^\top,$ where $\text{Bin}(d,p)$ is the binomial distribution with $d$ trials and success probability $p.$  
Each coordinate of $\omega_{i,t}$ denotes whether the information is transmitted or not.
Note that the generation of $\omega_t$ does not depend on data.
Define the active set $\Ac_{1,i,t} =\{k: [\omega_{i,t}]_k = 1\}$ and inactive set  $\Ac_{0,i,t} = \{k: [\omega_{i,t}]_k = 0\},$ and $\Ac_{1,t} = \cup_i \Ac_{1,i,t},$ $\Ac_{0,t} = \cup_i \Ac_{0,i,t}.$
Then, only active coordinates are released.
We perform this coordinate decomposition on $S(\d_t)$ since the adversaries can only infer sensitive data from the coordinates in $\Ac_{1,t}$, while the coordinates in $\Ac_{0,t}$ are private. 
In what follows, we compare our mechanism with two existing privacy protection techniques:

\smallskip
{\em 1) Difference from randomized response mechanism:}
In our sparsifier, the binomial vector $\z$ is used to determine the active  and inactive sets.
A similar mechanism is the randomized response (RR) mechanism \cite{dwork2014algorithmic,wang2016using}.
Our mechanism differs from RR as follows: 
i) The RR mechanism is a binary-response query, e.g., $\{\text{Yes},\text{No}\},$ while our query is sampling from $\mathbb{R}^d$;
ii) The RR mechanism is designed for the proportion estimation over several queries. 
However, in each iteration, only one query is available for the sparsifier.

\smallskip
{\em 2) Difference from sparse vector technique:}
In our privacy protection mechanism, due to the sparsity of the released vectors, part of information is protected.
This is similar to the idea of the sparse vector technique (SVT) \cite{dwork2009complexity,roth2010interactive,lyu2017understanding}. 
In SVT, a threshold value is chosen so that only the first $c$ queries above the threshold will be released. 
For privacy protection, two randomized procedures are used in SVT by adding noises on the threshold and queries.
The key difference between our method and SVT are: 
i) SVT is designed to select $c$ important queries (i.e. coordinates). 
However, in our algorithm, the sparsity of the released vector is random based on transmitted probability $p$.  
ii) In our method, the transmitted coordinates are amplified by a $(1/p)$-factor to ensure that the released vector is unbiased.
In contrast, it can be verified that the released vector in SVT is biased due to the lack of such an amplifying operation.

\begin{thm}[Privacy Guarantee] \label{Theorem: Privacy Guarantee}
Choose the variance of added Gaussian noise $\veta$ as $\sigma^2\ge 1/1.25$. 
Under Assumption~\ref{Assumption: function}, for any $\delta \in (0,1)$, 
the execution of SDM-DSGD algorithm with $T$ iterations is $(4\alpha\sum_{t=1}^{T}|\overline{\Ac_{1,t}}|(\tau G /\sqrt{d}m\sigma)^2+\epsilon/2,\delta)$-differentially-private,
where $\alpha = 2\log(1/\delta)/\epsilon+1$, $|\overline{\Ac_{1,t}}| = \max_i\{|\Ac_{1,i,t}|\}$ and $\tau$ is the subsampling rate for SGD.
\end{thm}

\begin{rem}
{\em The lower bound $\sigma^2 \geq 1/1.25$ follows from \cite{wang2018subsampled} to guarantee the privacy amplification under the subsampling. Theorem \ref{Theorem: Privacy Guarantee} shows that the sparsifier improves the differential privacy guarantee by a factor $\sum_{t=1}^{T}|\overline{\Ac_{1,t}}|/dT \le 1,$ which depends on $p.$ 
Hence, the smaller the value of $p$, the less information will be communicated, and the better privacy protection. 
Meanwhile, it can be seen that the privacy loss increases as the iteration number $T$ gets large, which is expected.
This is because, with fewer iterations, less information will be released, which implies a better privacy protection. 
However, fewer iterations cause a larger training loss in edge ML.
This leads to a {\em training-privacy trade-off}, which we will further analyze later.
Also, by inverting Theorem~\ref{Theorem: Privacy Guarantee}, we  have the following result:
}
\end{rem}

\begin{cor}\label{Cor: Privacy Guarantee}
Under the same conditions in Theorem \ref{Theorem: Privacy Guarantee} and let the subsampling rate be $1/m$ (i.e., each node subsamples one out of $m$ data), if the variance of added Gaussian noise $\veta$ is chosen as 
\begin{align}
\sigma^2 \ge \max\{\frac{8 \sum_{t=1}^{T}|\overline{\Ac_{1,t}}| G^2(2\log(1/\delta)+\epsilon)}{d\epsilon^2 m^4}, 
\frac{1}{1.25}\},
\end{align} then given the total iteration number $T,$ the SDM-DSGD algorithm is $(\epsilon,\delta)$-DP for any $\delta \in (0,1).$
\end{cor}

\smallskip
\textbf{2) Convergence Analysis for Training Loss:}
As shown in Eq.~(\ref{Eq: iterate update}), instead of directly optimizing $f(x;\D)$, our SDM-DSGD algorithm can be viewed as applying stochastic gradient descent on the Lyapunov function $\L_\gamma(\x;\D)$.
However, besides the random sampling noise, we also have the noises from the Gaussian masking and sparsification.
Note that the compression noise $\bepi$ is {\em dependent} on the sampling and added Gaussian masking noises, which significantly complicates our convergence analysis.
In what follows, we first quantify the optimization error of the sum output $\bar{x}_T=\sum\nolimits_{i=1}^{n}x_{i,T}.$

\begin{lemma}[Convergence]\label{Theorem: convergence}
Under Assumption \ref{Assumption: function}, fixing the variance of added Gaussian noise $\sigma^2 $ to be a constant, starting from $x_{i,0}=0$ $\forall i,$ and setting $\theta < 2p/(1-\lambda_n+\gamma L) \in (0,1),$ the iterates $x_{i,t}$ generated by Eq.~(\ref{Eq: iterate update}) satisfy:
\begin{align}\label{Eq: convergence error}
&\min_{t\in\{0,\cdots,T-1\}}\|\nabla f(\bx_t;\D)\|_2^2
\stackrel{}{\le} (\mathrm{I}) + (\mathrm{II}) + (\mathrm{III}) + (\mathrm{IV}),
\end{align}
where $(\mathrm{I}) = \frac{2C_1}{\theta\gamma T}$,
$(\mathrm{II}) = \frac{ 2LC_3}{n}\big(\frac{\gamma }{1-\beta}\big)^2$,
$(\mathrm{III}) = \frac{2\theta\gamma^2 LC_2}{n(1-\beta)}(\frac{1}{p}-1)+  \frac{L\theta\gamma C_2}{n^2p}$,
and $(\mathrm{IV}) = (\frac{2\gamma L}{n(1-\beta)} + \frac{L}{n^2})(\frac{1}{p}-1)\Big[\frac{2pn C_1}{\big(2p-(1-\lambda_n+\gamma L)\theta\big)T} +
\frac{(1-\lambda_n+\gamma L)\theta^2\gamma C_2}{2p-(1-\lambda_n+\gamma L)\theta}\Big]$.
In the above terms, $C_1 \triangleq f(0;\D) -  f(x^*_{\D};\D),$ $C_2 \triangleq n\tilde{\sigma}^2/m\tau\!+\!nd\sigma^2$ and $C_3 \triangleq (nG)^2+(nd\sigma)^2$ are constants, and $\tilde{\sigma}^2$ is the variance of the stochastic gradients.
\end{lemma}

\begin{rem}
{\em
There are four terms in the convergence error of SDM-DSGD in Eq.~(\ref{Eq: convergence error}): 
(I) is the common convergence error that goes to zero as $T$ and step-size $\theta\gamma$ increase;
(II) is the approximation error between the Lyapunov function $\L_\gamma(\x;\D)$ and $\f(\x;\D)$, which decreases with $\gamma$.
These two terms are similar to those in the convergence of DGD-based algorithms~\cite{zeng2018nonconvex,zhang2019compressed}; 
(III) and (IV) are the error terms introduced by the compression, random sampling, as well as the Gaussian masking noises.
The following simplified convergence rate result follows immediately from Lemma~\ref{Theorem: convergence}.
}
\end{rem}

\begin{cor}\label{Cor: Convergence}
Fixing the variance of Gaussian masking noise $\sigma^2 $ to be a constant, setting $\theta = \min\{p/(1-\lambda_n+\gamma L),p/2\}$, $\gamma = c\sqrt{n\log(T)/T}$, and $p \gg 1/\log(T)$, where $c$ is a constant, if the number of iterations satisfies $T > n^5/(1-\beta)^4$, then it holds that:
\begin{align}\label{Eq: order of convergence}
&\min_{t\in\{0,\cdots,T-1\}}\|\nabla f(\bx_t;\D)\|_2^2
\stackrel{}{= }  
O\Big(\sqrt{\frac{\log(T)}{nT}}\Big).
\end{align} 
\end{cor}

\begin{rem}
{\em
Several remarks for Corollary \ref{Cor: Convergence} are in order: 
1) The parameter $\theta$ is used to adjust the variance introduced by the sparsifier and could be set to a constant over the iterations, while the step-size $\gamma$ is required to be diminishing to control the sampling and Gaussian masking noises, as well as the approximation error;
2) The convergence rate is $O(\sqrt{\log(T)/nT}) =\tilde{O}(1/\sqrt{nT})$, which is approximately the same as the result in ~\cite{lian2017can} (ignoring the logarithm factor);
3) In the standard DSGD algorithm \cite{lian2017can}, to reach $\varepsilon$-accuracy, the communication complexity is $O(d/\varepsilon^2)$.
In contrast, the communication complexity of our algorithm is $O(1/\varepsilon^2\log(1/\varepsilon^3))$ by letting $p=1/d$.
Thus, our algorithm outperforms DSGD in overparameterized regime (i.e., large $d$);
4) The lower bound of $p$ is $1/\log(T).$ For example, with $10^4$ training iterations (a common setting for many deep learning training), the lower bound is approximately $0.1$, which is a small value;
5) The result in Corollary~\ref{Cor: Convergence} is based on two conditions: i) $\sigma^2$ is fixed over all iterations; and ii) the number of iterations $T$ is sufficiently large.
}
\end{rem}

Finally, by putting all aforementioned theoretical results together, we have the following key result for training-privacy trade-off:
\begin{thm}[Training-Privacy Trade-off] \label{Prop: privacy and convergence}
Under Assumption \ref{Assumption: function}, let $\sigma^2= 8 T G^2(2\log(1/\delta)+\epsilon)/m^4\epsilon^2$,
$\theta = \min\{p/(1-\lambda_n+\gamma L),p/2\},$ and $\gamma = c\sqrt{n\log(T)/T},$ where $c$ is a constant. 
If $T = m^4 \epsilon^2/ 20G^2\log(1/\delta) = O(m^{4})$, then the SDM-DSGD algorithm is $(\epsilon, \delta)$-DP and the convergence rate is 
\begin{align}\label{Eq: privacy convergence}
&\min_{t\in\{0,\cdots,T-1\}}\|\nabla f(\bx_t;\D)\|_2^2
\stackrel{}{= }  
\tilde{O}\Big(\frac{\sqrt{20G^2\log(1/\delta)}}{ \sqrt{n}m^2 \epsilon }\Big).
\end{align} 
\end{thm}

\begin{rem}
{\em 
Note that by letting $T = m^4 \epsilon^2/ 20G^2\log(1/\delta) = O(m^{4})$, we have that $\sigma^2>1/1.25$ over all iterations. 
Note also that the local sample size $m$ is usually much larger than the number of nodes, i.e., $m \gg n$, which implies that $ m^4 \epsilon^2/ 20G^2\log(1/\delta)>n^5/(1-\beta)^4.$ 
Hence, the convergence speed improvement still holds with $T = m^4 \epsilon^2/ 20G^2\log(1/\delta)$.
}
\end{rem}

\begin{figure}
\centering
\includegraphics[width=1\columnwidth]{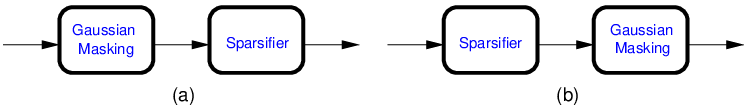}
\vspace{-.2in}
\caption{(a) SDM-DSGD scheme; (b) The alternative design.}
\label{fig:alt_design}
\vspace{-.2in}
\end{figure}

\smallskip
{\bf 3) Insights and Guidelines for Privacy and Communication Efficiency Co-Design:}
Note that in our algorithm, the Gaussian masking is applied before the sparse compression (see Figure~\ref{fig:alt_design}~(a). 
Thus, a fundamental and interesting question arises: {\em Is this a good design?}
To answer this question, consider the alternative design that {\em reverses} the Gaussian masking and sparsifier operations: we first perform sparsify operation on the local differential, and then apply Gaussian masking on those non-zero coordinates of the compressed local differential (see Figure~\ref{fig:alt_design}~(b)).
This alternative design can be mathematically written as:
\begin{align}\label{Eq: alternative design}
\left\{
\begin{aligned}
 \x_{t} &= \x_{t-1} + \underbrace{\big(S(\d_{t-1})+ \theta\gamma\tilde{\veta}_t \big)}_{\text{released message}}, \\
 \d_{t} &= (1-\theta)\x_{t} +\theta \Big(\tilde{\W} \x_{t} - \gamma \big(\g(x_{t};\zeta_{t}) \big)\Big) - \x_t,\\
\end{aligned}
\right.
\end{align}
where $[\tilde{\veta}_t]_{\Ac_{1,t}} \!\!\sim\! N(0,\sigma^2 \I)$ and $[\tilde{\veta}_t]_{\Ac_{0,t}} = \0;$ and the factor $\theta \gamma$ before $\tilde{\veta}_t$ is to make the result comparable.
For this alternative design, we have the following result:
\begin{prop}\label{Prop: alternative design}
Let the variance of Gaussian masking noise $\veta$ be chosen as $\sigma^2\ge 1/1.25$. 
Under Assumption~\ref{Assumption: function}, for any $\delta \in (0,1)$, the alternative design in (\ref{Eq: alternative design}) is $\big(4\alpha \sum_{t=1}^{T}|\overline{\Ac_{1,t}}|(\tau G )^2 /dm^2\sigma^2p^2 +\epsilon/2,\delta\big)$-DP,
where $\alpha \triangleq 2\log(1/\delta)/\epsilon-1$, $p$ is the transmit probability of the sparsifier, and $\tau$ is the subsampling rate.
\end{prop}
We can see that, in the ``$\epsilon$-part,'' the DP performance of our SDM-DSGD is smaller than that of the alternative design by a $(1/p^2)$-factor, and so our SDM-DSGD design is superior.
This difference is because in the alternative design, the sparse compression amplifies the $\ell_2$-sensitivity of the query by $1/p^2$.

\section{Experimental Evaluation}\label{Section: experiment}

In this section, we present experimental results of several nonconvex machine learning problems to evaluate the performance of our method.
In particular, we compare the communication cost and privacy-accuracy trade-off with two state-of-art algorithms:
\begin{list}{\labelitemi}{\leftmargin=1em \itemindent=0.em \itemsep=.2em}
	\item Decentralized SGD (DSGD)\cite{nedic2009distributed,yuan2016convergence,jiang2017collaborative}: Each node updates its local parameter as $x_{i,t+1} \!=\! \sum_{j \in \mathcal{N}_{i}} [\W]_{ij}x_{j,t}\! -\! \gamma g (x_{i,t}; \zeta_{i,t})$ with stochastic gradient $g (x_{i,t}; \zeta_{i,t})$ of random sample $\zeta_{i,t}$, and exchange the uncompressed local parameter $x_i$ with its neighbors. 
	\item Differential Compressed Decentralized SGD (DC-DSGD) \cite{tang2018communication}: This algorithm also communicates compressed local differentials and estimating neighbors' copies. However, in the local copy updating step, DC-DSDG does not have tuning parameter $\theta$ (or can be viewed as fixing $\theta = 1$ in our SDM-DSGD).
\end{list}
In our experiment, we choose the transmit probability $p$ from  $\{1,0.5,0.2\}$. 
We set $\theta = 1$ for $p=1$ and $0.5,$ which is corresponding to DSGD and DC-DSGD, respectively. 
However, for $p=0.2,$ the algorithm does not converge if we choose $\theta=1$, i.e., DC-DSGD fails under $p=0.2$ (See Figure \ref{Fig: counterexample}). 
Thus, we set $\theta = 0.6$ for the case when $p=0.2$, which is corresponding to our algorithm.

\begin{figure}
\centering
\begin{tabular}{@{}cc@{}}
\includegraphics[width=0.49\linewidth]{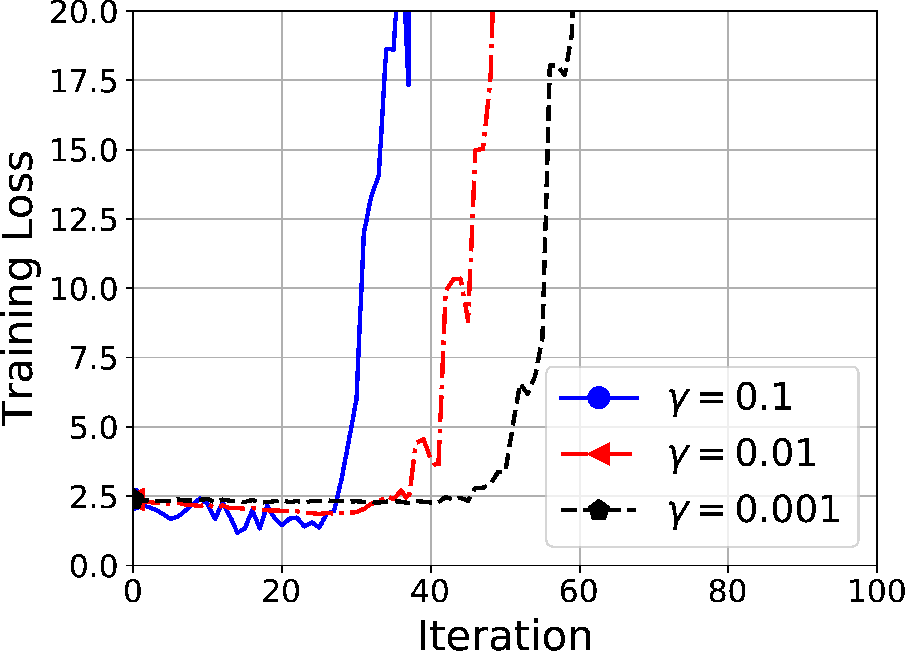}&
\includegraphics[width=0.48\linewidth]{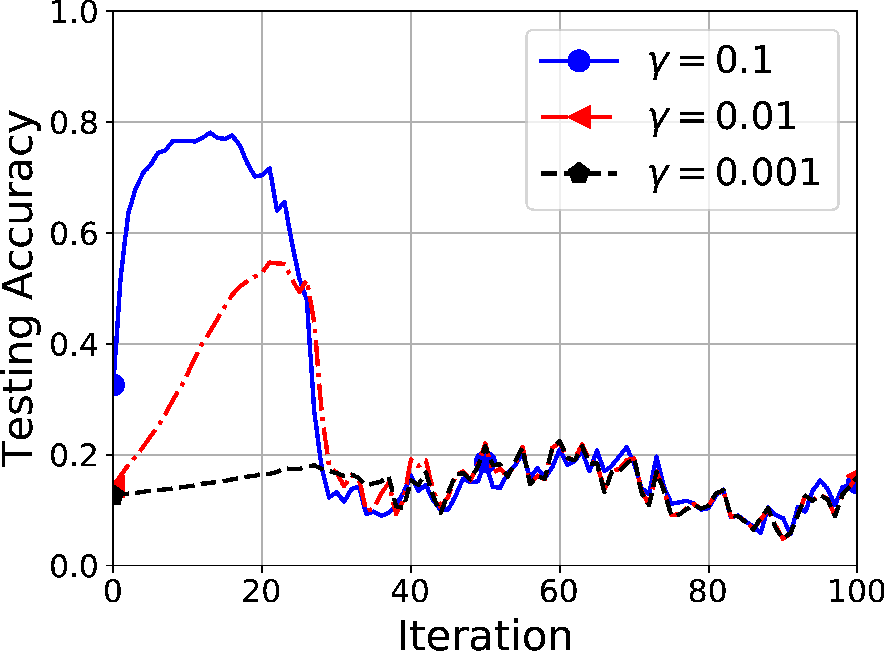} \\
\multicolumn{2}{c}{(a) MLR on MNIST.} \medskip
\\
\includegraphics[width=0.49\linewidth]{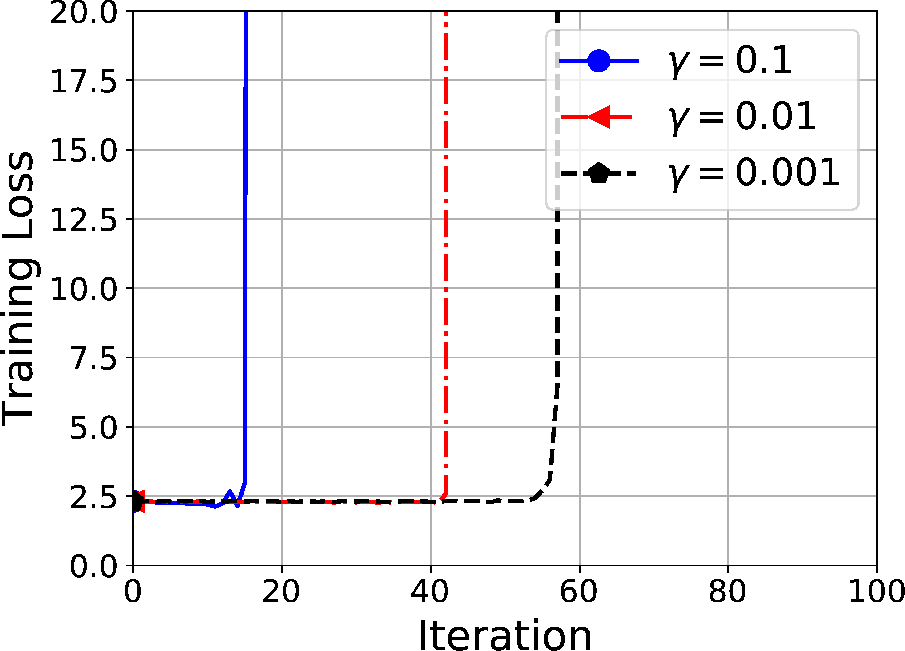} &
\includegraphics[width=0.48\linewidth]{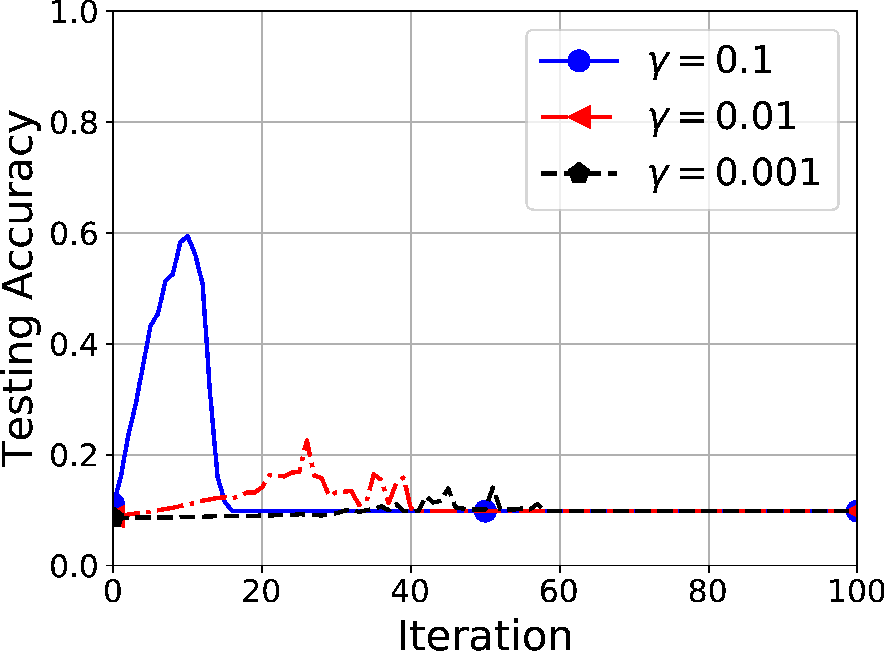}\\
\multicolumn{2}{c}{(b) CNN on MNIST.}
\end{tabular}
\vspace{-.1in}
\caption{Two examples where DC-DSGD diverges: $p\!=\!0.2$ and the step-size $\gamma$ is chosen from $\{0.1,0.01,0.001\}$.}\label{Fig: counterexample}
\vspace{-.2in}
\end{figure}

\smallskip
\textbf{Dataset and Learning Models:} 
We adopt all three algorithms to solve a variety of nonconvex learning problems over MNIST and CIFAR-10 datasets.
For MNIST, we apply the multi-class logistic regression (MLR) and convolutional neural network (CNN) classifiers.
The adopted CNN model has two convolutional layers (size $3 \times 3 \times 16$), each is followed by a max-pooling layer with size $2\times 2$, and then a fully connected layer. 
The ReLU activation is used for the two convolutional layers and the ``softmax'' activation is applied at the output layer. 
For CIFAR-10, we apply the above CNN model and the ResNet20 model.
The batch size is 64 for both MLR and the CNN classifiers on MNIST.
The batch size is 128 and 32 for the CNN and ResNet20 classifiers on CIFAR-10, respectively.

\smallskip
\textbf{Network Model:} 
We use a network with 50 nodes.
Similar to \cite{reisizadeh2019exact,reisizadeh2019robust}, the communication graph $\mathcal{G}$ is generated by the Erd$\ddot{\text{o}}$s-R$\grave{\text{e}}$nyi graph with edge connectivity $p_c=0.35.$ The network concensus matrix is chosen as $\W = \I - \frac{2}{3\lambda_{\text{max}}(\mathbf{L})} \mathbf{L},$ where $\mathbf{L}$ is the Laplacian matrix of $\mathcal{G}$, and $\lambda_{\text{max}}(\mathbf{L})$ denotes the largest eigenvalue of $\mathbf{L}$. 

\smallskip
\textbf{Procedure for Privacy:} 
Note that privacy protection is not considered in the original DSGD and DC-DSGD methods.
To have a fair comparison, we add the same Gaussian noise $N(\0,\I)$ to the stochastic gradients, so that all algorithms are privacy-preserving.
To control the object function's $\ell_{2}$-sensitivity to $\x$, we adopt a modified gradient clipping technique \cite{abadi2016deep}: $\text{Clip}([g]_i) = \text{sign}([g]_i) \max\{|[g]_i|,C\},$ $\forall g\in \mathbb{R}^d.$ 
With this clipping, each coordinate of the gradient is bounded by $C$ in magnitude. Here, we set $C = 5$. 
In our experiment, we keep track of the privacy loss based on Theorem~\ref{Theorem: Privacy Guarantee}.

\begin{figure}
\centering
\begin{tabular}{@{}cc@{}}
\includegraphics[width=0.49\linewidth]{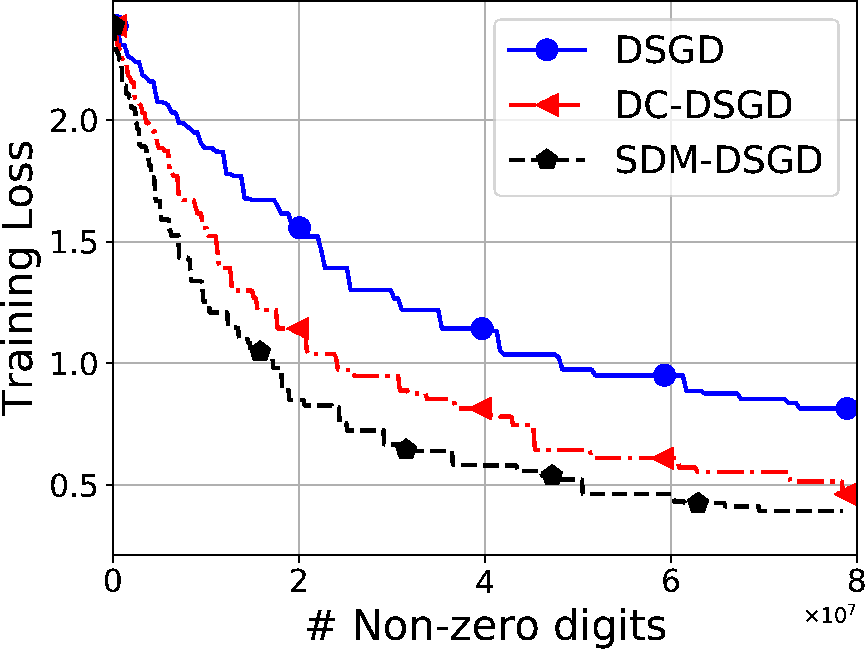}&
\includegraphics[width=0.49\linewidth]{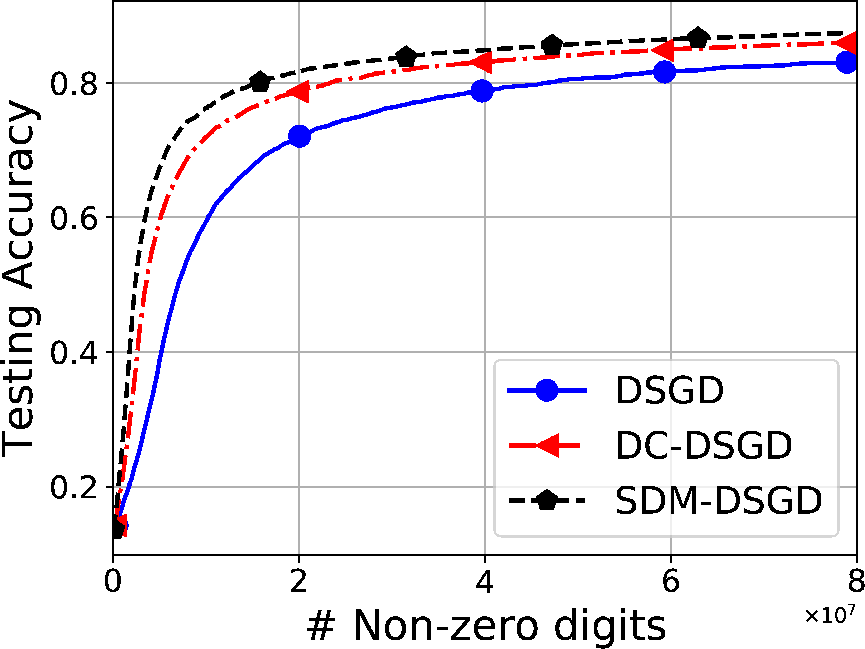} \\
\multicolumn{2}{c}{(a) MLR on MNIST.} \\
\includegraphics[width=0.49\linewidth]{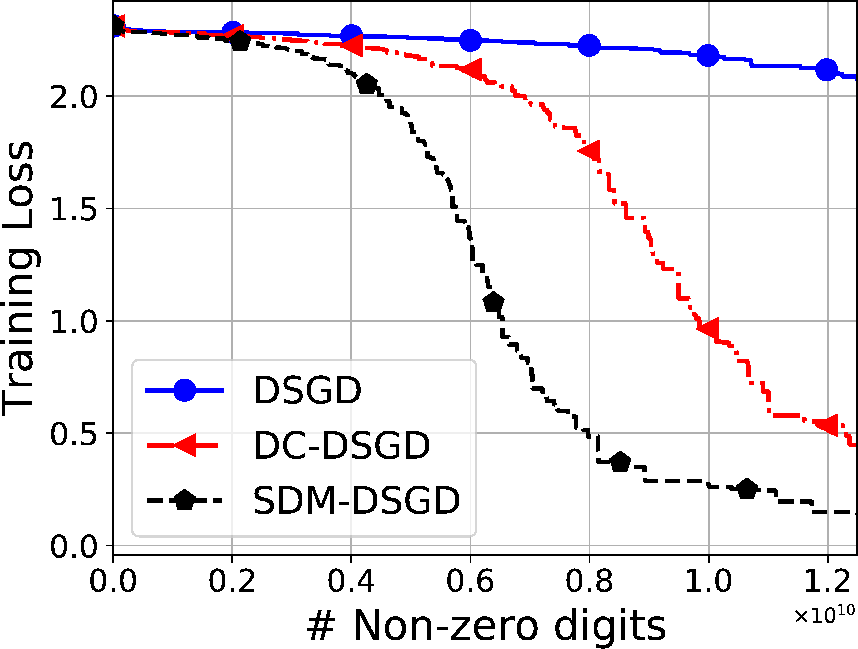} &
\includegraphics[width=0.49\linewidth]{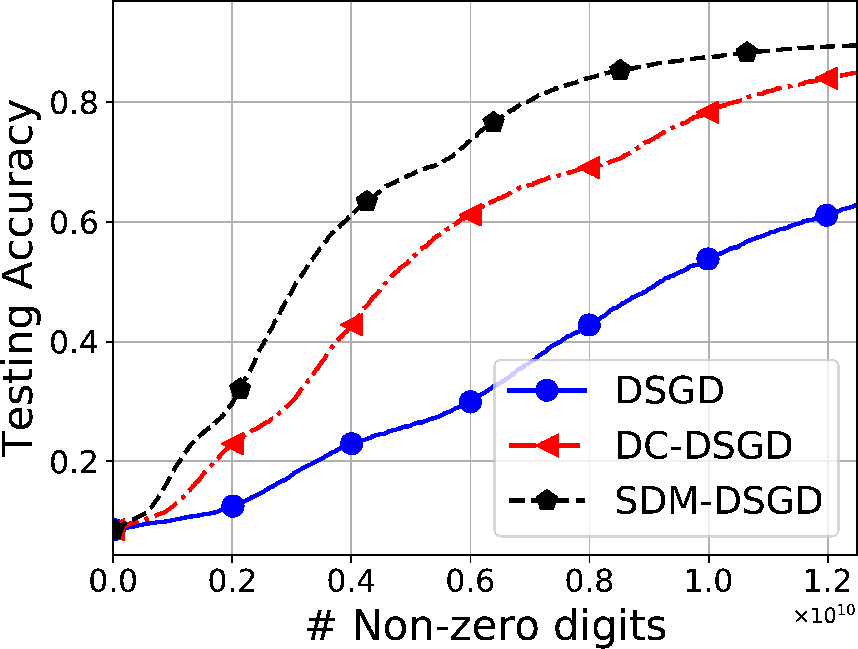}\\
\multicolumn{2}{c}{(b) CNN on MNIST.}\\
\includegraphics[width=0.49\linewidth]{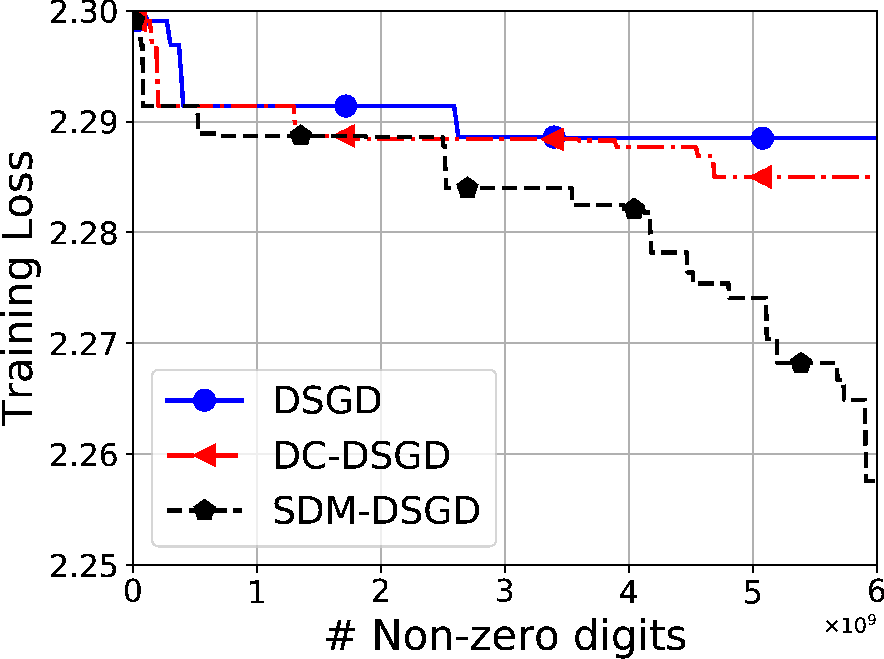}&
\includegraphics[width=0.49\linewidth]{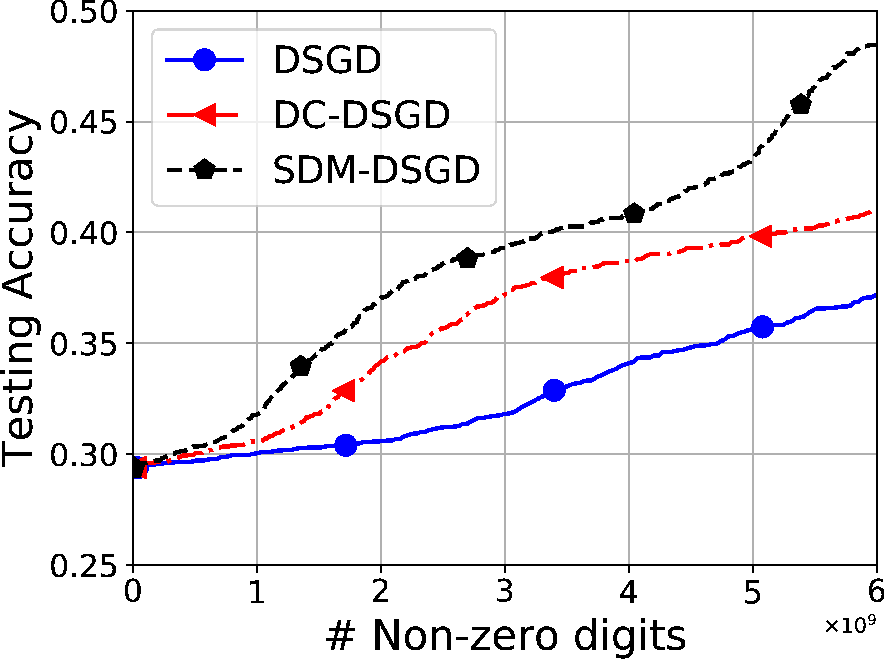} \\
\multicolumn{2}{c}{(c) CNN on CIFAR-10.} \\
\includegraphics[width=0.49\linewidth]{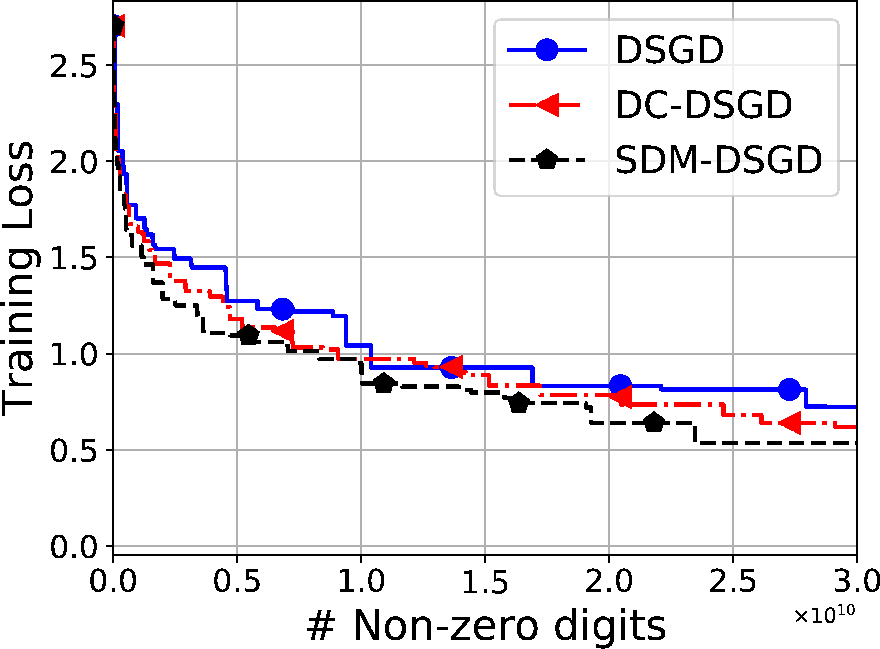} &
\includegraphics[width=0.49\linewidth]{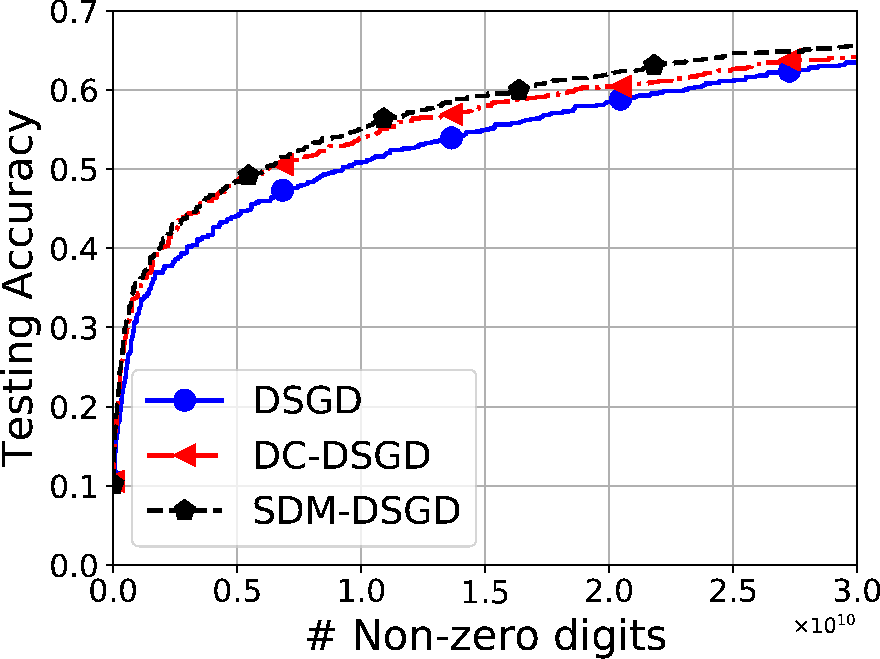}\\
\multicolumn{2}{c}{(d) ResNet20 on CIFAR-10.}\\
\end{tabular}
\vspace{-.1in}
\caption{Results of objective loss (left) and testing accuracy (right) with models trained by different algorithms.}\label{Fig: exp_results}
\vspace{-.2in}
\end{figure}

\begin{table}
\caption{The results of testing accuracy with models trained by different algorithms with $(\epsilon,\delta=10^{-5})$-DP guarantee.}\label{Table: Testing Accuracy}
\vspace{-.16in}
\begin{center}
\begin{tabular}{p{2.0cm}<{\centering} p{1.5cm}<{\centering} p{1.5cm}<{\centering} p{1.5cm}<{\centering}  }
\hline
\hline
\multicolumn{4}{c}{MLR on MNIST}\\
\hline
$\epsilon (\times 10^{-3})$ & $1.0$ & $2.0$ & $5.0$  \\
\hline
DSGD & 0.1422 & 0.1956 & 0.6324 \\
DC-DSGD  & 0.1621 & 0.2959 & 0.7408 \\
{\bf SDM-DSGD}  & {\bf 0.1880} & {\bf 0.4296} & {\bf 0.7810} \\
\hline
\hline
\multicolumn{4}{c}{CNN on MNIST}\\
\hline
$\epsilon (\times 10^{-2})$ & $2.0$ & $5.0$ & $10.0$  \\
\hline
DSGD & 0.0886 & 0.1059 & 0.2570  \\
DC-DSGD & 0.0917 & 0.1442 & 0.5265 \\
{\bf SDM-DSGD}  & {\bf 0.0960} & {\bf 0.2150} & {\bf 0.6728} \\
\hline
\hline
\multicolumn{4}{c}{CNN on CIFAR-10}\\
\hline
$\epsilon (\times 10^{-2})$ & $5.0$ & $10.0$ & $20.0$  \\
\hline
DSGD & 0.2960 & 0.3036 & 0.3544  \\
DC-DSGD & 0.2991 & 0.3292 & 0.3964 \\
{\bf SDM-DSGD}  & {\bf 0.3013} & {\bf 0.3570} & {\bf 0.4296} \\
\hline
\hline
\multicolumn{4}{c}{ResNet20 on CIFAR-10}\\
\hline
$\epsilon (\times 10^{-2})$ & $2.0$ & $5.0$ & $10.0$  \\
\hline
DSGD & 0.3265 & 0.4631 & 0.5735  \\
DC-DSGD & 0.3324 & 0.4922 & 0.5957 \\
{\bf SDM-DSGD}  & {\bf 0.3470} & {\bf 0.5079} & {\bf 0.6099} \\
\hline
\hline
\end{tabular}
\end{center}
\vspace{-.2in}
\end{table}

\smallskip
\textbf{Numerical Results:}
We illustrate the results of training loss and testing accuracy with respect to communication costs in Figure~\ref{Fig: exp_results}.
We compute the total non-zero digits (i.e. the non-sparsified digits) communicated in the each iteration, which is used to measure the communication cost in the training.
In the left-hand-side figures in Figure~\ref{Fig: exp_results}, we show the training loss vs. the amount of the non-zero digits.
In the right-hand-side figures in Figure~\ref{Fig: exp_results}, we show the testing accuracy vs the amount of the non-zero digits.
We can see that under the same amount of non-zero digits, our SDM-DSGD has the fastest convergence speed and the best testing accuracy: 
in the case of training CNN on MNIST, with $6\times 10^9$ non-zero digits,  SDM-DSGD's testing accuracy is at $80\%$, while those of DC-DSGD and DSGD are $60\%$ and less than $40\%$, respectively.

For privacy loss, we summarize the results in Table~\ref{Table: Testing Accuracy}.
We can see that under the same $\delta$, the testing accuracy is increasing as the privacy budget gets large (i.e., large $\epsilon$): 
in the case of MLR on MNIST, with $\epsilon$ increasing from $1\times 10^3$ to $5 \times 10^3$, 
testing accuracy is increasing from $19\%$ to $78\%.$ 
This is because with larger privacy budget, more information is allowed to be released, which leads to a better training.
Meanwhile, under the same privacy budget $\epsilon,$ our SDM-DSGD consistently has a higher accuracy compared with the other two algorithms:
our algorithm improves the accuracy at least $4\%$ with MLR and $15\%$ with CNN on MNIST. 

\section{Conclusion}\label{Section: conclusion}

In this paper, we proposed the SDM-DSGD algorithm to improve both data privacy and communication efficiency in distributed edge learning.
In our SDM-DSGD algorithm, we proposed to exchange the sparsified differentials between the computation nodes and develop a ``generalized'' computing scheme for local updates.
We theoretically and numerically showed that by doing so, the proposed algorithm converges with a small sparsification probability $p$ compared with the state-of-the-art DC-DSGD method.
Moreover, we considered the protection data privacy by injecting a Gaussian masking noise.
We studied the interaction between the sparsification and the Gaussian masking mechanism and showed that the ``randomize-then-sparisify'' is the preferred approach.
Finally, we established the privacy-accuracy tradeoff theoretically.
Through extensive experiments, we showed that SDM-DSGD outperforms existing algorithms.
Our results advance the state-of-the-art of communication efficiency and data privacy in distributed edge learning.

\bibliographystyle{plain}
\bibliography{reference}

\appendix
\section{Proof of main results}

\subsection{Proof of Theorem \ref{Theorem: Privacy Guarantee}}
\begin{proof}
To prove the privacy guarantee, we first states the following definitions and related lemmas.

\begin{defn}[$\ell_2$-Sensitivity]\label{Def: l2 sensitivity}
The $\ell_2$-Sensitivity is defined as the maximum change in the $\ell_2$-Norm of the function value $f(\cdot)$ on two adjacent datasets $\D$ and $\D^\prime$:
\begin{align}
\Delta(f) = \max_{\D,\D^\prime} \|f(\D)-f(\D^\prime)\|_2,
\end{align}
where $\D$ and $\D^\prime$ are two adjacent datasets if only one element are different.
\end{defn}

\begin{defn}[R$\ee$nyi Differential Privacy (RDP)~\cite{mironov2017renyi}]\label{Def: RDP}
A randomized mechanism $\Mcc:\Dc^N \rightarrow \mathbb{R}^d$ is $(\alpha,\rho)$-R$\ee$nyi Differential Privacy $\big((\alpha,\rho)\text{-RDP}\big)$, if for any two adjacent dataset $\D,$ $\D^\prime \in \Dc^N,$ it holds that
\begin{align}
D_\alpha(\Mcc(\D)\|\Mcc(\D^\prime)) = \frac{1}{\alpha-1}\log\mathbb{E}\Big(\frac{\Mcc(\D)}{\Mcc(\D^\prime)}\Big)^\alpha \le \rho,
\end{align}
where the expectation is taken over $\Mcc(\D^\prime).$
$D_\alpha(\cdot\|\cdot)$ is also known as R$\ee$nyi divergence.
\end{defn}

\begin{lem}[RDP of Gaussian Mechanism~\cite{wang2018subsampled}]\label{Lemma: RDP Gaussian}
Consider the mechanism $\Mcc = f(\D)+\veta,$ with the function $f:\Dc\rightarrow \R^d$ and Gaussian noise $\veta\sim N(0,\sigma^2\I_d).$
\begin{enumerate} [topsep=1pt, itemsep=-.1ex, leftmargin=.3in]
\item[i).]\label{Lemma: RDP Gaussian1} The mechanism $\Mcc$ is $(\alpha,\alpha \Delta^2(f)/(2\sigma^2)$-RDP, where $\Delta(f)$ is the $\ell_2$-Sensitivity of $f;$
\item[ii).]\label{Lemma: RDP Gaussian2} If the mechanism $\Mcc$ is applied to a subset of samples using uniform sampling without replacement and $\sigma^2 \ge 1.1/25,$ then $\Mcc$ is $(\alpha,\alpha\tau^2 \Delta^2(f)/\sigma^2)$-RDP, where $\tau$ is the sampling rate. 
\end{enumerate}
\end{lem}

\begin{lem}[Sequential Composibility of RDP~\cite{mironov2017renyi}]\label{Lemma: RDP composibility}
Consider $f: \Dc^N \rightarrow \mathbb{R}^{d_1}$ and $g: \Dc^N\times\mathbb{R}^{d_1} \rightarrow \mathbb{R}^{d_2}.$ If $f$ is $(\alpha,\rho_1)$-RDP and $g$ is $(\alpha,\rho_2)$-RDP, then the mechanism $(X,Y)$ satisfies $(\alpha,\rho_1+\rho_2)$-RDP, where $X\sim f(\D)$ and $Y \sim g(X,\D).$
\end{lem}

\begin{lem}[From RDP to $(\epsilon,\delta)$-DP~\cite{mironov2017renyi}]\label{Lemma: RDP2DP}
If $f$ is an $(\alpha, \rho)$-RDP mechanism, then it also satisfies $\big(\rho+\log(1/\delta)/(\alpha-1),\delta\big)$-DP for any $\delta \in (0,1).$
\end{lem}

With the above definition and lemmas, we provide the following privacy guarantee for our algorithm in the following.
Our proof is inspired by \cite{wang2019dp}.
Given the active set at $t$th iteration $\Ac_{1,t},$ and $\Ac_{1,i,t}$ respective to the $i$th node, the updating equation (\ref{Eq: iterate update}) can be rewritten as:
\begin{align}
[\x_{t+1}]_{\Ac_{1,t}} &= [\x_t]_{\Ac_{1,t}} - [(\theta\nabla\L_\gamma(\x_t;\zeta_t) + \theta\gamma \veta_t)/p]_{\Ac_{1,t}}\\
& = [\x_t]_{\Ac_{1,t}} - \theta\big([\nabla\L_\gamma(\x_t;\zeta_t)]_{\Ac_{1,t}} + \gamma [\veta_t]_{\Ac_{1,t}}\big)/p
\end{align}
Thus, we need to analyze the privacy gaurantee of the above SGD updating with noise $[\veta_t]_{\Ac_{1,t}},$ of which each coordinate is from $N(0,\sigma^2).$
Given dataset $\D,$ consider the following mechanism $\hat{\Mcc}_t = [\nabla \L_\gamma(\x_t;\D)]_{\Ac_{1,t}}+\gamma [\veta_t]_{\Ac_{1,t}}$ with the query $\q_t = [\nabla \L_\gamma(\x_t;\D)]_{\Ac_{1,t}}$.
With the adjacent datasets $\D$ and $\D^\prime,$ the $\ell_2$-sensitivity of $\q_t$ is 
\begin{align}
\Delta(\q_t) 
&
\stackrel{}{=} \|[\nabla \L_\gamma(\x_t;\D)]_{\Ac_{1,t}} - [\nabla \L_\gamma(\x_t;\D^\prime)]_{\Ac_{1,t}} \|_2 \notag\\
&
\stackrel{}{=} \|[(\tilde{\W}\x_t - \gamma\nabla \f(\x_{t};\D_i)) - (\tilde{\W}\x_t -\gamma\nabla \f(\x_{t};\D^\prime_i))]_{\Ac_{1,t}} \|_2 \notag\\
&
\stackrel{}{=} \gamma\sum\nolimits_{i=1}^{n}\|[\nabla f(x_{i,t};\D_i) -\nabla f(x_{i,t};\D^\prime_i)]_{\Ac_{1,i,t}} \|_2 \notag\\
&
\stackrel{(a)}{=} \frac{\gamma}{m}\|[\nabla f(x_{i,t}|\zeta_{i,j}) -\nabla f(x_{i,t};\zeta^\prime_{i,j})]_{\Ac_{1,i,t}} \|_2 \notag\\
&
\stackrel{}{=} \frac{\gamma}{m}\sqrt{\sum\nolimits_{k \in \Ac_{1,i,t}}([\nabla f(x_{i,t};\zeta_{i,j})]_k-[\nabla f(x_{i,t};\zeta_{i,j}^\prime)]_k)^2} \notag\\
&
\stackrel{(b)}{\le} \frac{2\gamma\sqrt{|\Ac_{1,i,t}|}G}{\sqrt{d}m} 
\le \frac{2\gamma\sqrt{|\overline{\Ac_{1,t}}|}G}{\sqrt{d}m}
\end{align}
where (a) by assuming that the only different data is $\zeta_{i,j}$ and $\zeta^\prime_{i,j}$ in the $i$th node; 
(b) by the coordinate-wise $G/\sqrt{d}$-Lipschitz of the function $f(\cdot),$ and $|\overline{\Ac_{1,t}}| = \max_i\{|\Ac_{1,i,t}|\} \le d.$
Thus, based on Lemma \ref{Lemma: RDP Gaussian} i), with $\veta_t \sim N(0,\sigma^2\I_{nd}),$ the mechanism $\hat{\Mcc}_t$ satisfies $\big(\alpha,2\alpha|\overline{\Ac_{1,t}}|(G /\sqrt{d}m\sigma)^2\big)$-RDP\footnote{Note that the mechanism $\hat{\Mcc}_t$ adds the Gaussian noise $\gamma[\veta_t]_{\Ac_{1,t}} \sim N(0,\gamma^2\sigma^2\I_{\Ac_{1,i,t}})$.}. 
Then for the mechanism $\Mcc_t = [\nabla \L_\gamma(\x_t;\zeta_t)]_{\Ac_{1,t}}+\gamma [\veta_t]_{\Ac_{1,t}},$ which is equivalent to applying $\hat{\Mcc}_t$ to a subset of random sample $\zt_t$, according to Lemma \ref{Lemma: RDP Gaussian} ii), $\Mcc_t$ satisfies $\big(\alpha,4\alpha|\overline{\Ac_{1,t}}|(\tau G /\sqrt{d}m\sigma)^2\big)$-RDP with $\sigma^2\ge 1/1.25.$ 
So set $\alpha = 2\log(1/\delta)/\epsilon+1$ and with Lemma \ref{Lemma: RDP2DP}, we have $\Mcc_t$ satisfies $(4\alpha|\overline{\Ac_{1,t}}|(\tau G /\sqrt{d}m\sigma)^2+\epsilon/2,\delta)$-DP with $\sigma^2\ge 1/1.25.$

Next, we derive the privacy guarantee over $T$ iterations.
By Lemma \ref{Lemma: RDP composibility}, with $T$ iterations, i.e. sequentially composition of $\{\Mcc_t\}_{t=1}^{T}$, the algorithm output $\x_T$ satisifies $\big(\alpha,\sum_{t=1}^{T}4\alpha|\overline{\Ac_{1,t}}|(\tau G /\sqrt{d}m\sigma)^2\big)$-RDP.
With Lemma \ref{Lemma: RDP2DP}, $\x_T$ satisfies $(4\alpha\sum_{t=1}^{T}|\overline{\Ac_{1,t}}| (\tau G/\sqrt{d}m\sigma)^2+\epsilon/2,\delta)$-DP when $\sigma^2\ge 1/1.25.$
\end{proof}

\subsection{Proof of Proposition \ref{Prop: alternative design}} 
\begin{proof}
Under this design, we have the updating:
\begin{align}
[\x_{t+1}]_{\Ac_{1,t}} &= [\x_t]_{\Ac_{1,t}} - [\theta\nabla\L_\gamma(\x_t;\zeta_t)/p +\theta\gamma \tilde{\veta}_t]_{\Ac_{1,t}}.
\end{align}
Consider the mechanism $\hat{\Mcc}_t = [\theta\nabla \L_\gamma(\x_t;\D)/p]_{\Ac_{1,t}}+ \theta\gamma[\tilde{\veta}_t]_{\Ac_{1,t}} $ with the query $\q_t = [\theta\nabla \L_\gamma(\x_t;\D)/p]_{\Ac_{1,t}}$.
With the adjacent datasets $\D$ and $\D^\prime,$ we have the $\ell_2$-sensitivity of $\q_t$ is 
\begin{align}
\Delta(\q_t) 
&
\stackrel{}{=} \frac{\theta}{p} \|[\nabla \L_\gamma(\x_t;\D)]_{\Ac_{1,t}} - [\nabla \L_\gamma(\x_t;\D^\prime)]_{\Ac_{1,t}} \|_2 \notag\\
&
\stackrel{}{=} \frac{\theta}{p} \|[(\tilde{\W}\x_t - \gamma\nabla \f(\x_{t};\D_i)) - (\tilde{\W}\x_t -\gamma\nabla \f(\x_{t};\D^\prime_i))]_{\Ac_{1,t}} \|_2 \notag\\
&
\stackrel{}{=} \frac{\theta\gamma}{p}\|\sum\nolimits_{i=1}^{n}[\nabla f(x_{i,t};\D_i) -\nabla f(x_{i,t};\D^\prime_i)]_{\Ac_{1,i,t}} \|_2 \notag\\
&
\stackrel{(a)}{=} \frac{\theta\gamma}{mp}\|[\nabla f(x_{i,t};\zeta_{i,j}) -\nabla f(x_{i,t};\zeta^\prime_{i,j})]_{\Ac_{1,i,t}} \|_2 \notag\\
&
\stackrel{}{=} \frac{\theta\gamma}{mp}\sqrt{\sum\nolimits_{k \in \Ac_{1,i,t}}([\nabla f(x_{i,t};\zeta_{i,j})]_k-[\nabla f(x_{i,t};\zeta_{i,j}^\prime)]_k)^2} \notag\\
&
\stackrel{(b)}{\le} \frac{2\theta\gamma\sqrt{|\Ac_{1,i,t}|}G}{\sqrt{d}mp}
\le \frac{2\theta\gamma\sqrt{|\overline{\Ac_{1,t}}|}G}{\sqrt{d}mp}
\end{align}
where (a) by assuming that the only different data is $\zeta_{i,j}$ and $\zeta^\prime_{i,j}$ in the $i$th node; 
(b) by the coordinate-wise $G/\sqrt{d}$-Lipschitz of the function $f(\cdot).$
Thus, based on Lemma \ref{Lemma: RDP Gaussian} i), with $[\tilde{\veta_t}]_{\Ac_{1,t}} \sim N(0,\sigma^2\I),$ the mechanism $\hat{\Mcc}_t$ satisfies $\big(\alpha,2\alpha|\overline{\Ac_{1,t}}|(G/\sqrt{d}m\sigma p)^2\big)$-RDP. 
Define the mechanism $\Mcc_t = 
[\theta\nabla \L_\gamma(\x_t;\zeta_t)/p]_{\Ac_{1,t}}+ [\tilde{\veta}_t]_{\Ac_{1,t}} .$
With the similar derivation, we have $\Mcc_t$ satisfies $(4\alpha|\overline{\Ac_{1,t}}| (\tau G )^2 /dm^2\sigma^2 p^2+\epsilon/2,\delta)$-DP with $\sigma^2\ge 1/1.25.$
Hence with $T$ iterations, the algorithm satisfies $(4\alpha \sum_{t=1}^{T}|\overline{\Ac_{1,t}}|(\tau G )^2 /dm^2\sigma^2 p^2+\epsilon/2,\delta)$-DP with $\sigma^2\ge 1/1.25.$
\end{proof}

\subsection{Proof of Theorem \ref{Theorem: convergence}}

\begin{proof}
First, we give the following useful lemma.
\begin{lem}\label{Lemma: sparse dt}
Under the same conditions in Lemma \ref{Theorem: convergence}, at $t$th iteration, the random sparsified output $S(\d_t)$ has:
\begin{enumerate} [topsep=1pt, itemsep=-.1ex, leftmargin=.3in]
\item[i).] First Moment: $\Eb[S(\d_t)|\x_t] = -\theta\nabla \L_\gamma(\x_t;\D);$
\item[ii).] Second Moment: $\Eb[\|S(\d_t)\|_2^2|\x_t] \le \frac{\theta^2}{p}\|\nabla \L_\gamma(\x_t;\D)\|_2^2 +  \frac{(\theta\gamma)^2}{p}(\frac{n\tilde{\sigma}^2}{m\tau}+nd\sigma^2).$
\end{enumerate}
\end{lem} 
\begin{proof}
i). For the first moment,
\begin{align}
&
\Eb[S(\d_t)|\x_t] 
=
\Eb[\Eb[S(\d_t)|\d_t]|\x_t]= \Eb[\d_t|\x_t]\notag\\ 
&
=
\Eb[- \theta(\nabla\L_\gamma(\x_t;\zeta_t)+\gamma \veta_t)|\x_t]\notag\\
&
= -\theta\nabla \L_\gamma(\x_t;\D)
\end{align}
ii). For the second moment, 
\begin{align}
&
\Eb[\|S(\d_t)\|_2^2|\x_t] 
= 
\|\Eb[S(\d_t)]|\x_t\|_2^2 +\text{Var}[S(\d_t)|\x_t] \notag\\
&
\stackrel{(a)}{=} 
\theta^2\|\nabla \L_\gamma(\x_t;\D)\|_2^2+ \Eb[\text{Var}[S(\d_t)|\d_t]|\x_t] + \text{Var}[\Eb[S(\d_t)|\d_t]|\x_t] \notag\\
&
\stackrel{(b)}{=} 
\theta^2\|\nabla \L_\gamma(\x_t;\D)\|_2^2+ (\frac{1}{p} -1)\Eb[\|\d_t\|_2^2|\x_t] + \text{Var}[\d_t|\x_t] \notag\\
&
\stackrel{(c)}{\le} 
\theta^2\|\nabla \L_\gamma(\x_t;\D)\|_2^2 + (\frac{1}{p} -1)\theta^2\Eb [\|\nabla\L_\gamma(\x_t;\zeta_t)+\gamma \veta_t\|_2^2|\x_t] \notag\\
&
+ (\theta\gamma)^2(\frac{n\tilde{\sigma}^2}{m\tau}+nd\sigma^2)\notag\\
&
\stackrel{(d)}{=} 
\theta^2\|\nabla \L_\gamma(\x_t;\D)\|_2^2 + (\frac{1}{p} -1)\theta^2\Eb [\|\nabla\L_\gamma(\x_t;\zeta_t)\|_2^2+\|\gamma \veta_t\|_2^2|\x_t]\notag\\
&
+ (\theta\gamma)^2(n\frac{\tilde{\sigma}^2}{m\tau}+nd\sigma^2)\notag\\
&
\stackrel{}{\le} 
\theta^2\|\nabla \L_\gamma(\x_t;\D)\|_2^2 + (\frac{1}{p} -1)\theta^2[\|\nabla\L_\gamma(\x_t;\D)\|_2^2 \notag\\
&
+\gamma^2(\frac{n\tilde{\sigma}^2}{mn\tau}+d\sigma^2)]+ (\theta\gamma)^2(\frac{\tilde{\sigma}^2}{m\tau}+nd\sigma^2)\notag\\
&
\stackrel{}{=} 
\frac{\theta^2}{p}\|\nabla \L_\gamma(\x_t;\D)\|_2^2 +  \frac{(\theta\gamma)^2}{p}(\frac{n\tilde{\sigma}^2}{m\tau}+nd\sigma^2)
\end{align}
where (a) is by the Eve's law;
(b) is from the properies of the sparsifier in Section \ref{Section: preliminary};
(c) is by $\d_t = - \theta(\nabla\L_\gamma(\x_t;\zeta_t)+\gamma \veta_t),$ at each node the subsampling rate is $\tau;$
(d) is because the randomness of the subsampling and the Gaussian mechanism are independent.
\end{proof}

Step 1: Define the filtration $\Fc_t = \sigma\langle \x_1,\cdots,\x_t \rangle.$
Note that the Lyapunov function $\L_{\gamma}(\x;\D)$ has $(1-\lambda_n+\gamma L)$-Lipschitz gradient.
Thus, we have 
\begin{align}
&\L_\gamma(\x_{t+1};\D) \notag\\
&
\le \L_\gamma(\x_{t};\D) + \langle \nabla \L_{\gamma}(\x_{t};\D),\x_{t+1}-\x_{t} \rangle 
 + \frac{(1-\lambda_n+\gamma L)}{2}\|\x_{t+1}-\x_{t}\|^2 \notag\\
&
= \L_\gamma(\x_{t};\D) + \langle \nabla \L_{\gamma}(\x_{t};\D), S(\d_{t}) \rangle + \frac{(1-\lambda_n+\gamma L)}{2}\|S(\d_{t})\|^2
\end{align}
Take conditional expectation at both sides:
\begin{align}
&\Eb[\L_\gamma(\x_{t+1})|\Fc_t]\notag\\
&
\le
\L_\gamma(\x_{t};\D) + 
\langle \nabla \L_{\gamma}(\x_{t};\D), \Eb[S(\d_{t})|\Fc_t] \rangle + 
\frac{(1-\lambda_n+\gamma L)}{2}\Eb[\|S(\d_{t})\|^2|\Fc_t] \notag\\
& 
= 
\L_\gamma(\x_{t};\D) -
\theta \| \nabla \L_{\gamma}(\x_{t};\D)\|^2_2 
+
\frac{(1-\lambda_n+\gamma L)}{2}[ \frac{\theta^2}{p}\|\nabla \L_\gamma(\x_t;\D)\|_2^2 + 
\notag\\
&
\frac{(\theta\gamma)^2}{p}(\frac{n\tilde{\sigma}^2}{m\tau}+d\sigma^2)]
\notag\\
& 
\le 
\L_\gamma(\x_{t}) + 
\big(\frac{(1-\lambda_n+\gamma L)\theta^2}{2p}-\theta\big) \| \nabla \L_{\gamma}(\x_{t};\D)\|^2_2 \notag\\
&
+
\frac{(1-\lambda_n+\gamma L)(\theta\gamma)^2}{2p}(\frac{n\tilde{\sigma}^2}{m\tau}+nd\sigma^2)
\end{align}
Thus, by setting $2p\theta-(1-\lambda_n+\gamma L)\theta^2 > 0,$ i.e. $\theta < 2p/(1-\lambda_n+\gamma L),$ we have the following descent inequality:
\begin{align}\label{Eq: descent inequality L}
&\big(2p\theta-(1-\lambda_n+\gamma L)\theta^2\big) \| \nabla \L_{\gamma}(\x_{t};\D)\|^2_2\notag\\
&\le 
2p(\L_\gamma(\x_{t})-\Eb[\L_\gamma(\x_{t+1})|\Fc_t]) +
(1-\lambda_n+\gamma L)(\theta\gamma)^2(\frac{n\tilde{\sigma}^2}{m\tau}+nd\sigma^2).
\end{align}
Telescope the inequalities from $t=0$ to $T$, it holds that:
\begin{align}
&\big(2p\theta-(1-\lambda_n+\gamma L)\theta^2\big) \sum\nolimits_{t=0}^{T} \Eb\| \nabla \L_{\gamma}(\x_{t};\D)\|^2_2\notag\\
&
\le 
2p(\L_\gamma(\x_{0};\D)-\Eb[\L_\gamma(\x_{T+1};\D)])  \notag\\
&
+
(1-\lambda_n+\gamma L)(\theta\gamma)^2(\frac{n\tilde{\sigma}^2}{m\tau}+nd\sigma^2)(T+1)
\end{align}
Because of the two facts that $\L_\gamma(\x_{T+1};\D) \ge \gamma \sum\nolimits_{i=1}^{n}f(x_{i,T+1};\D_i) \ge \gamma \sum\nolimits_{i=1}^{n}f(x^*_{\D};\D_i)$ and $\L_\gamma(\x_0;\D) = \gamma \sum\nolimits_{i=1}^{n}f(\0;\D_i),$ it holds that: 
\begin{align}
&\big(2p\theta-(1-\lambda_n+\gamma L)\theta^2\big) \sum\nolimits_{t=0}^{T} \Eb\| \nabla \L_{\gamma}(\x_{t};\D)\|^2_2\notag\\
&\le 
2p\gamma\big(\sum\nolimits_{i=1}^{n}f(\0;\D_i)-\sum\nolimits_{i=1}^{n}f(x^*_{\D};\D_i)\big) \notag\\
&
+ 
(1-\lambda_n+\gamma L)(\theta\gamma)^2(\frac{n\tilde{\sigma}^2}{m\tau}+nd\sigma^2)(T+1).
\end{align}
Thus, we have
\begin{align}
\sum\nolimits_{t=0}^{T} \Eb\| \nabla \L_{\gamma}(\x_{t};\D)\|^2_2
&\le 
\frac{2pn\gamma C_1}{\big(2p\theta-(1-\lambda_n+\gamma L)\theta^2\big)} \notag\\
&
+
\frac{(1-\lambda_n+\gamma L)\theta\gamma^2(T+1)C_2}{2p-(1-\lambda_n+\gamma L)\theta}.
\end{align}
where $C_1 = f(\0;\D)-f(x^*_{\D};\D)$ and $C_2 = n\tilde{\sigma}^2/m\tau+nd\sigma^2$ are two constants.

Step 2: In the following, we provide the bound for $\|\x_t-\xb_t\| = \|\big(\I_{nd} - (\frac{1}{n}\1_n\1_n^\top)\otimes \I_d\big)\x_t\|,$ where $\xb_t = \big((\frac{1}{n}\1_n\1_n^\top)\otimes \I_d\big)\x_t.$
For notation convenience, we define $\Q = \I_{nd} - (\frac{1}{n}\1_n\1_n^\top)\otimes \I_d,$ $\f(\x;\zeta)=\sum\nolimits_{i=1}^{n}f(x_i;\zeta_i).$
From the updating (\ref{Eq: iterate update}), it holds that:
\begin{align}
\x_{t}
&
= \x_{t-1} - \theta(\nabla\L_\gamma(\x_{t-1};\zeta_{t-1}) + \gamma \veta_{t-1}) + \bepi_{t-1} \notag\\
&
= \x_{t-1} - \theta\big( (\I-\tilde{\W})\x_{t-1} + \gamma \nabla \f(\x_{t-1};\zeta_{t-1}) + \gamma \veta_{t-1}\big) + \bepi_{t-1}  \notag\\
& 
= 
\big((1-\theta) \I +\theta\tilde{\W}\big)\x_{t-1} -
\theta\gamma  \nabla \f(\x_{t-1};\zeta_{t-1}) - \theta\gamma \veta_{t-1}+ \bepi_{t-1}.
\end{align}
It can be seen that the updating is the stochastic DGD updating with a mixed concensus matrix $\tilde{\W}_\theta = (1-\theta) \I +\theta\tilde{\W},$ which is also doubly stochastic with $\theta \in (0,1)$.
Thus, it holds that starting $\x_0 = \0,$
\begin{align}
\x_{t}
&
= 
\tilde{\W}_\theta\x_{t-1} -
\theta\gamma  \nabla \f(\x_{t-1};\zeta_{t-1}) - \theta\gamma \veta_{t-1}+ \bepi_{t-1}\notag\\
&
= 
\sum\nolimits_{s = 0 }^{t-1} \tilde{\W}_\theta^{t-1-s} 
\big(-
\theta\gamma  \nabla \f(\x_{s};\zeta_{s}) - \theta\gamma \veta_{s}+ \bepi_{s}\big),
\end{align}
then due to the rows sums and columns sum of $\tilde{\W}_\theta$ are $1,$ it holds that 
\begin{align}
&\Q\x_t 
= 
\Q\sum\nolimits_{s = 0 }^{t-1} \tilde{\W}_\theta^{t-1-s} 
\big(-
\theta\gamma  \nabla \f(\x_{s};\zeta_{s}) - \theta\gamma \veta_{s}+ \bepi_{s}\big)\notag\\
&
= 
\sum\nolimits_{s = 0 }^{t-1} \big(\tilde{\W}_\theta^{t-1-s} - (\frac{1}{n}\1_n\1_n^\top)\otimes \I_d\big)
\big(-
\theta\gamma  \nabla \f(\x_{s};\zeta_{s}) - \theta\gamma \veta_{s}+ \bepi_{s}\big),
\end{align}
which results to
\begin{align}
&
\|\Q\x_t \|_2^2 
\notag\\
&
=
2\|\sum\nolimits_{s = 0 }^{t-1} \big(\tilde{\W}_\theta^{t-1-s} - (\frac{1}{n}\1_n\1_n^\top)\otimes \I_d\big)
\big(-
\theta\gamma  \nabla \f(\x_{s};\zeta_{s}) - \theta\gamma \veta_{s} \big)\|^2_2  \notag\\
&
+ 
2\|\sum\nolimits_{s = 0 }^{t-1} \big(\tilde{\W}_\theta^{t-1-s} - (\frac{1}{n}\1_n\1_n^\top)\otimes \I_d\big)
\bepi_{s}\|^2_2 \notag\\
 &
=
2\|\sum\nolimits_{s = 0 }^{t-1} \big(\tilde{\W}_\theta^{t-1-s} - (\frac{1}{n}\1_n\1_n^\top)\otimes \I_d\big)
\big(-
\theta\gamma  \nabla \f(\x_{s};\zeta_{s}) - \theta\gamma \veta_{s} \big)\|^2_2  \notag\\
&
+ 
\sum\nolimits_{s = 0 }^{t-1} 2\|\big(\tilde{\W}_\theta^{t-1-s} - (\frac{1}{n}\1_n\1_n^\top)\otimes \I_d\big)
\bepi_{s}\|^2_2 
\notag\\
&
+ 
\sum\nolimits_{s,s^\prime = 0,s\neq s^\prime }^{t-1}
2\langle
\big(\tilde{\W}_\theta^{t-1-s} - (\frac{1}{n}\1_n\1_n^\top)\otimes \I_d\big)
\bepi_{s}
,\notag\\
&
\big(\tilde{\W}_\theta^{t-1-s^\prime} - (\frac{1}{n}\1_n\1_n^\top)\otimes \I_d\big)
\bepi_{s^\prime}
 \rangle
\end{align}
Take the expectation at the both sides,
\begin{align}
&\frac{1}{2}\Eb[\|\Q\x_t \|_2^2]
\notag\\
&
\stackrel{(a)}{=}
\Eb[\|\sum\nolimits_{s = 0 }^{t-1} \big(\tilde{\W}_\theta^{t-1-s} - (\frac{1}{n}\1_n\1_n^\top)\otimes \I_d\big)
\big(-
\theta\gamma  \nabla \f(\x_{s};\zeta_{s}) - \theta\gamma \veta_{s} \big)\|^2_2]  \notag\\
&
~~
+ 
\sum\nolimits_{s = 0 }^{t-1} \Eb[\|\big(\tilde{\W}_\theta^{t-1-s} - (\frac{1}{n}\1_n\1_n^\top)\otimes \I_d\big)
\bepi_{s}\|^2_2 ] \notag\\
&
\le
\Eb[\|\sum\nolimits_{s = 0 }^{t-1} \big(\tilde{\W}_\theta^{t-1-s} - (\frac{1}{n}\1_n\1_n^\top)\otimes \I_d\big)
\big(-
\theta\gamma  \nabla \f(\x_{s}|\zeta_{s}) - \theta\gamma \veta_{s} \big)\|^2_2]  \notag\\
&
~~
+ 
\sum\nolimits_{s = 0 }^{t-1}  \beta_\theta^{2(t-1-s)} (\frac{1}{p}-1)\Eb\|\theta(\nabla \L(\x_s|\zeta_s) + \gamma\veta_s)\|_2^2\notag\\
&
\le
(\theta\gamma)^2
\sum\nolimits_{s = 0 }^{t-1} \beta_\theta^{t-1-s}\sum\nolimits_{s^\prime = 0 }^{t-1} \beta_\theta^{t-1-s^\prime}
\Eb[
  \|\nabla \f(\x_{s};\zeta_{s}) +  \veta_{s}\|^2_2 ]^{\frac{1}{2}}
   \notag\\
&
 \Eb[\| \nabla \f(\x_{s^\prime};\zeta_{s^\prime}) +  \veta_{s^\prime} \|_2^2]^{\frac{1}{2}}
+ 
(\frac{1}{p}-1)\Eb\|\theta(\nabla \L(\x_s;\zeta_s) + \gamma\veta_s)\|_2^2\notag\\
&
\stackrel{(b)}{\le}
(\theta\gamma)^2
\sum\nolimits_{s = 0 }^{t-1} \beta_\theta^{t-1-s}\sum\nolimits_{s^\prime = 0 }^{t-1} \beta_\theta^{t-1-s^\prime}
((nG)^2+(nd\sigma)^2) \notag\\
&
+ 
(\frac{1}{p}-1)\Eb\|\theta(\nabla \L(\x_s;\zeta_s) + \gamma\veta_s)\|_2^2\notag\displaybreak[4]\\
& 
\stackrel{}{\le }
\frac{(\theta\gamma)^2((nG)^2+(nd\sigma)^2)}{(1-\beta_\theta)^2}
+ 
\sum\nolimits_{s = 0 }^{t-1} \beta_\theta^{2(t-1-s)} \theta^2(\frac{1}{p}-1)\Eb[\|\nabla \L(\x_s;\D)\|^2 
 \notag\\
&
+ 
\gamma^2 (\frac{n\tilde{\sigma}^2}{m\tau} + nd\sigma^2)] \notag\\
& 
\stackrel{}{\le }
\frac{(\theta\gamma)^2((nG)^2+(nd\sigma)^2)}{(1-\beta_\theta)^2}
+ 
\frac{\theta^2\gamma^2 C_2}{1-\beta_\theta^2}(\frac{1}{p}-1)  \notag\\
&
+
\sum\nolimits_{s = 0 }^{t-1} \beta_\theta^{2(t-1-s)} \theta^2(\frac{1}{p}-1)\Eb[\|\nabla \L(\x_s;\D)\|^2]  \notag\\
& 
\stackrel{(c)}{\le }
\Big(\frac{\theta\gamma }{1-\beta_\theta}\Big)^2C_3
+ 
\frac{\theta^2\gamma^2 C_2}{1-\beta_\theta}(\frac{1}{p}-1)  \notag]\\
&
+ 
\sum\nolimits_{s = 0 }^{t-1} \beta_\theta^{2(t-1-s)} \theta^2(\frac{1}{p}-1)\Eb\|\nabla \L(\x_s;\D)\|^2\notag\\
& 
\stackrel{(d)}{\le }
\Big(\frac{\gamma }{1-\beta}\Big)^2C_3
+ 
\frac{\theta\gamma^2 C_2}{1-\beta}(\frac{1}{p}-1) 
 \notag\\
&
+
\sum\nolimits_{s = 0 }^{t-1} \beta_\theta^{2(t-1-s)} \theta^2(\frac{1}{p}-1)\Eb[\|\nabla \L(\x_s;\D)\|^2 ] 
\end{align}
where $\beta_\theta = \max \{|\lambda_2(\W_\theta)|,|\lambda_n(\W_\theta)|\}$ with $\W = (1-\theta) \I + \theta\W$ and (a) is because of $\Eb[\bepi_t]=\0;$ (b) is because the function $\f(\x;\zeta)$ is coordinately $G/\sqrt{d}$-Lipschitz and hence 
$\Eb[\|\nabla \f(\x;\zeta)+\eta^2\|_2^2] = \Eb[\|\nabla \f(\x;\zeta)\|_2^2+\|\eta^2\|_2^2] \le (nG)^2+(nd\sigma)^2;$ 
(c) is because of $\beta_\theta \in (0,1)$ and $C_3 =
(nG)^2+(nd\sigma)^2;$ (d) is from Lemma \ref{Lemma: beta/theta}.

\begin{lem}\label{Lemma: beta/theta}
Given $\theta\in (0,1),$ it holds
\begin{align}
\frac{1}{1-\beta_\theta} \le \frac{1}{\theta(1-\beta)}
\end{align}
with $\beta_\theta \!=\! \max\{|\lambda_2(\W_\theta)|,|\lambda_n(\W_\theta)|\}$ and $\beta \!=\! \max\{|\lambda_2(\W)|,|\lambda_n(\W)|\}.$
\end{lem}
\begin{proof}
First, for $1/(1-\beta),$ according to the definition:
\begin{align}
&\beta = \max\{|\lambda_2(\W)|,|\lambda_n(\W)|\} \notag\\
\Rightarrow~&1-\beta = \min\{1-|\lambda_2(\W)|,1-|\lambda_n(\W)|\}\notag\\
\Rightarrow~&1/(1-\beta)  = \max\{1/(1-|\lambda_2(\W)|),1/(1-|\lambda_n(\W)|)\}.
\end{align}
Then for $1/(1-\beta_\theta),$ note that $\W_\theta = (1-\theta)\I +\theta\W,$ which implies $\lambda_i(\W_\theta) = (1-\theta) + \theta\lambda_i(\W).$
So $\beta_\theta = \max\{|(1-\theta) + \theta\lambda_2(\W)|,|(1-\theta) + \theta\lambda_n(\W)|\}.$  
Note that for any $\lambda \in (-1,1]$ and $\theta \in (0,1)$ it holds that $|(1-\theta) + \theta\lambda| \le (1-\theta) + \theta|\lambda|,$ thus,
\begin{align}
&\beta_\theta \le \max\{(1-\theta) + \theta|\lambda_2(\W)|,(1-\theta) + \theta|\lambda_n(\W)|\} \notag\\
\Rightarrow~&1-\beta \ge \min\{\theta-\theta|\lambda_2(\W)|,\theta-\theta|\lambda_n(\W)|\}\notag\\
\Rightarrow~&1/(1-\beta_\theta)  \le \max\{1/\theta(1-|\lambda_2(\W)|),1/\theta(1-|\lambda_n(\W)|)\}\notag\\
\Rightarrow~&1/(1-\beta_\theta)  \le 1/\theta(1-\beta).
\end{align}
\end{proof}

Step 3: Note that $\x_t = \tilde{\W}_\theta\x_{t-1} -
\theta\gamma  \nabla \f(\x_{t-1}|\zeta_{t-1}) - \theta\gamma \veta_{t-1}+ \bepi_{t-1},$ which implies that 
\begin{align}
\bx_t &= \bx_{t-1} - \frac{1}{n}\sum\nolimits_{i=1}^{n}[\theta\gamma\nabla f(x_{i,t-1} ;\zeta_{i,t-1}) - \theta\gamma \eta_{i,t-1}+ \epsilon_{i,t-1}]\notag\\
&= \bx_{t-1} - \frac{\theta\gamma}{n}\sum\nolimits_{i=1}^{n}[\nabla f(x_{i,t-1}|\zeta_{i,t-1}) - \eta_{i,t-1}+ \epsilon_{i,t-1}/\theta\gamma]\notag\\
&= \bx_{t-1} - \frac{\theta\gamma}{n}\sum\nolimits_{i=1}^{n}\nabla \tf(x_{i,t-1}|\Bc_{i,t-1}),
\end{align}
where $\nabla \tf(x_{i,t}|\Bc_{i,t}) = \nabla f(x_{i,t};\zeta_{i,t}) - \eta_{i,t}+ \epsilon_{i,t}/\theta\gamma,$ and $\Bc_{i,t}=\sigma\langle \zeta_{i,t},\eta_{i,t},\epsilon_{i,t}/\theta\gamma\rangle.$

Consider $f(\bx_t;\D)=\frac{1}{n}\sum_{i=1}^{n}f(\bx_{t};\D_i),$ by the $L$-Lipschitz continuous gradient, it holds:
\begin{align}
&
f(\bx_{t+1};\D)
\le 
f(\bx_{t};\D) + \langle \nabla f(\bx_t;\D),\bx_{t+1}-\bx_{t} \rangle + \frac{L}{2}\|\bx_{t+1}-\bx_{t}\|_2^2 \notag\\
&
\le
f(\bx_{t};\D) - \langle \nabla f(\bx_t;\D),\frac{\theta\gamma}{n}\sum\nolimits_{i=1}^{n}\nabla \tf(x_{i,t}|\Bc_{i,t}) \rangle \notag\\
&
+ \frac{L}{2}\|\frac{\theta\gamma}{n}\sum\nolimits_{i=1}^{n}\nabla \tf(x_{i,t}|\Bc_{i,t})\|_2^2 \notag\\
&
\le
f(\bx_{t};\D) - \langle \nabla f(\bx_t;\D), \theta\gamma\nabla f(\bx_t;\D) + \frac{\theta\gamma}{n}\sum\nolimits_{i=1}^{n}\nabla \tf(x_{i,t}|\Bc_{i,t}) \notag\\
&
- \theta\gamma\nabla f(\bx_t;\D) \rangle
+
 \frac{L}{2}\|\theta\gamma\nabla f(\bx_t;\D) + \frac{\theta\gamma}{n}\sum\nolimits_{i=1}^{n}\nabla \tf(x_{i,t}|\Bc_{i,t}) 
 \notag\\
&
- \theta\gamma\nabla f(\bx_t;\D)\|_2^2 \notag\\
&
\le
f(\bx_{t};\D) - \theta\gamma \|\nabla f(\bx_t;\D)\|_2^2 - \langle \nabla f(\bx_t;\D), \frac{\theta\gamma}{n}\sum\nolimits_{i=1}^{n}\nabla \tf(x_{i,t}|\Bc_{i,t}) \notag\\
&
-
 \theta\gamma\nabla f(\bx_t;\D) \rangle  
+ \frac{L}{2}[\|\theta\gamma\nabla f(\bx_t;\D)\|_2^2 + \|\frac{\theta\gamma}{n}\sum\nolimits_{i=1}^{n}\nabla \tf(x_{i,t}|\Bc_{i,t}) 
\notag\\
&
- \theta\gamma\nabla f(\bx_t;\D)\|_2^2  
+
 2\langle\theta\gamma\nabla f(\bx_t;\D),  \frac{\theta\gamma}{n}\sum\nolimits_{i=1}^{n}\nabla \tf(x_{i,t}|\Bc_{i,t})\notag\\
&
 - \theta\gamma\nabla f(\bx_t;\D)\rangle] 
\end{align}
Take conditional expectation, it holds:
\begin{align}
&
\Eb[f(\bx_{t+1};\D)|\Fc_t]\notag\\
&\stackrel{(a)}{\le}
 f(\bx_{t};\D) - \theta\gamma \|\nabla f(\bx_t;\D)\|_2^2 - \langle \nabla f(\bx_t;\D), \frac{\theta\gamma}{n}\sum\nolimits_{i=1}^{n}\nabla f(x_{i,t};\D_{i}) \notag\\
&
- \theta\gamma\nabla f(\bx_t;\D) \rangle 
+ \frac{L}{2}[\|\theta\gamma\nabla f(\bx_t;\D)\|_2^2 + \Eb[\|\frac{\theta\gamma}{n}\sum\nolimits_{i=1}^{n}\nabla \tf(x_{i,t}|\Bc_{i,t})
\notag\\
&
 - \theta\gamma\nabla f(\bx_t;\D)\|_2^2|\Fc_{t}] 
+
 2\langle\theta\gamma\nabla f(\bx_t;\D),  \frac{\theta\gamma}{n}\sum\nolimits_{i=1}^{n}\nabla f(x_{i,t};\D_{i}) 
 \notag\\
&
- \theta\gamma\nabla f(\bx_t;\D)\rangle] \notag\\
 &
=
f(\bx_{t};\D) - (\theta\gamma -\frac{L(\theta\gamma)^2}{2})\|\nabla f(\bx_t;\D)\|_2^2  \notag\\
&
+ (\theta\gamma-L(\theta\gamma)^2)
\langle \nabla f(\bx_t;\D),  \nabla f(\bx_t;\D)-\frac{1}{n}\sum\nolimits_{i=1}^{n}\nabla f(x_{i,t};\D_{i}) \rangle  \notag\\
&
+ \frac{L(\theta\gamma)^2}{2}\Eb[\|\frac{1}{n}\sum\nolimits_{i=1}^{n}\nabla \tf(x_{i,t}|\Bc_{i,t}) - \nabla f(\bx_t;\D)\|_2^2|\Fc_{t}] \notag\\
 &
\stackrel{(b)}{\le}
f(\bx_{t};\D) - (\theta\gamma -\frac{L(\theta\gamma)^2}{2})\|\nabla f(\bx_t;\D)\|_2^2 
+
\frac{\theta\gamma-L(\theta\gamma)^2}{2} \times \notag\\
&
[\| \nabla f(\bx_t;\D)\|_2^2  
+ \|\frac{1}{n}\sum\nolimits_{i=1}^{n}\nabla f(x_{i,t};\D_{i}) - \nabla f(\bx_t;\D)\|_2^2 ] \notag\displaybreak[3]\\
&
~~
+\frac{L(\theta\gamma)^2}{2}\Eb[\|\frac{1}{n}\sum\nolimits_{i=1}^{n}\nabla \tf(x_{i,t}|\Bc_{i,t}) - \nabla f(\bx_t;\D)\|_2^2|\Fc_{t}] \notag\\
 &
=
f(\bx_{t};\D) - \frac{\theta\gamma}{2}\|\nabla f(\bx_t;\D)\|_2^2  \notag\\
&
+
\frac{\theta\gamma-L(\theta\gamma)^2}{2}
\|\frac{1}{n}\sum\nolimits_{i=1}^{n}\nabla f(x_{i,t};\D_{i}) - \nabla f(\bx_t;\D)\|_2^2  \notag\\
&
~~
+\frac{L(\theta\gamma)^2}{2}\underbrace{\Eb[\|\frac{1}{n}\sum\nolimits_{i=1}^{n}\nabla \tf(x_{i,t}|\Bc_{i,t}) - \nabla f(\bx_t;\D)\|_2^2|\Fc_{t}]}_{(A)} 
\end{align}
where (a) is because $\Eb[\nabla \tf(x_{i,t}|\Bc_{i,t})|\Fc_{t}] = \nabla f(x_{i,t};\D_i);$ (b) is by $2\langle \a, \b \rangle \le \|\a\|_2^2 + \|\b\|_2^2.$ 
Now we give the bound for $(A):$ Note that 
\begin{align}
&
\Eb[\|\frac{1}{n}\sum\nolimits_{i=1}^{n}\nabla \tf(x_{i,t}|\Bc_{i,t}) - \nabla f(\bx_t;\D)\|_2^2|\Fc_{t}] \notag\\
&
=
\Eb[\|\frac{1}{n}\sum\nolimits_{i=1}^{n}\nabla \tf(x_{i,t}|\Bc_{i,t})-\frac{1}{n}\sum\nolimits_{i=1}^{n}\nabla f(x_{i,t};\D)\notag\\
&
~~
+\frac{1}{n}\sum\nolimits_{i=1}^{n}\nabla f(x_{i,t};\D) - \nabla f(\bx_t;\D)\|_2^2|\Fc_{t}]\notag\\
&
\stackrel{(a)}{=}
\Eb[\|\frac{1}{n}\sum\nolimits_{i=1}^{n}\nabla \tf(x_{i,t}|\Bc_{i,t})-\frac{1}{n}\sum\nolimits_{i=1}^{n}\nabla f(x_{i,t};\D)\|_2^2|\Fc_{t}]\notag\\
&
~~
+
\|\frac{1}{n}\sum\nolimits_{i=1}^{n}\nabla f(x_{i,t};\D) - \nabla f(\bx_t;\D)\|_2^2.
\end{align}
where (a) is because $\Eb[\frac{1}{n}\sum\nolimits_{i=1}^{n}\nabla \tf(x_{i,t}|\Bc_{i,t})|\Fc_t] = \frac{1}{n}\sum\nolimits_{i=1}^{n}\nabla f(x_{i,t};\D).$
Recall the definition that $\nabla \tf(x_{i,t}|\Bc_{i,t}) = \nabla f(x_{i,t};\zeta_{i,t}) - \eta_{i,t}+ \epsilon_{i,t}/\theta\gamma$, hence,
\begin{align}
&\Eb[\|\frac{1}{n}\sum\nolimits_{i=1}^{n}\nabla \tf(x_{i,t}|\Bc_{i,t})-\frac{1}{n}\sum\nolimits_{i=1}^{n}\nabla f(x_{i,t};\D)\|_2^2|\Fc_{t}] \notag\\
&
\stackrel{(a)}{=}
\frac{1}{n^2}\sum\nolimits_{i=1}^{n}
\Eb[\|\nabla \tf(x_{i,t}|\Bc_{i,t})-\nabla f(x_{i,t};\D)\|_2^2|\Fc_{t}]\notag\\
&
=
\frac{1}{n^2}\sum\nolimits_{i=1}^{n}
\Eb[\|\nabla f(x_{i,t};\zeta_{i,t}) - \eta_{i,t}+ \epsilon_{i,t}/\theta\gamma-\nabla f(x_{i,t};\D)\|_2^2|\Fc_{t}] \notag\\
&
\stackrel{(b)}{=}
\frac{1}{n^2}\sum\nolimits_{i=1}^{n}
\Eb[\|\nabla f(x_{i,t};\zeta_{i,t})-\nabla f(x_{i,t};\D)\|_2^2 
\notag\\
&
+\| \eta_{i,t}\|_2^2+ \|\epsilon_{i,t}/\theta\gamma\|_2^2|\Fc_{t}] \notag\\
&
\le
\frac{1}{n}
[\frac{\tilde{\sigma}^2}{m\tau} + d\sigma^2] + \big(\frac{1}{n\theta\gamma}\big)^2\Eb[\|\bepi_{t}\|_2^2|\Fc_{t}]\notag\displaybreak[3]\\
&
\le
\frac{1}{n}
[\frac{\tilde{\sigma}^2}{m\tau} + d\sigma^2] + \big(\frac{1}{n\theta\gamma}\big)^2\Eb[(\frac{1}{p}-1)\|\theta\L_\gamma(\x_t|\zeta_t)+\theta\gamma\veta_t\|_2^2|\Fc_{t}]\notag\\
&
\le
\frac{1}{n}
[\frac{\tilde{\sigma}^2}{m\tau} + d\sigma^2] + \big(\frac{1}{n\gamma}\big)^2(\frac{1}{p}-1)[\|\L_\gamma(\x_t;\D)\|_2^2+\frac{\gamma^2\tilde{\sigma}^2n}{m\tau} + \gamma^2\sigma^2nd]\notag\\
&
=
\big(\frac{1}{n\gamma}\big)^2(\frac{1}{p}-1)\|\L_\gamma(\x_t;\D)\|_2^2 + \frac{1}{np}(\frac{\tilde{\sigma}^2}{m\tau} + \sigma^2d)
\end{align}
where (a) is because the noise is independent across $i,$ and (b) is because the expectation of the three terms are zero.
Thus, we have 
\begin{align}
&
\Eb[f(\bx_{t+1};\D)|\Fc_t] \notag\\
&
\le
f(\bx_{t};\D) - \frac{\theta\gamma}{2}\|\nabla f(\bx_t;\D)\|_2^2  \notag\\
&
+
\frac{\theta\gamma-L(\theta\gamma)^2}{2}
\|\frac{1}{n}\sum\nolimits_{i=1}^{n}\nabla f(x_{i,t};\D_{i}) - \nabla f(\bx_t;\D)\|_2^2  \notag\\
&
+\frac{L(\theta\gamma)^2}{2}\Eb[\|\frac{1}{n}\sum\nolimits_{i=1}^{n}\nabla \tf(x_{i,t}|\Bc_{i,t}) - \nabla f(\bx_t;\D)\|_2^2|\Fc_{t}] \notag\\
&
\le
f(\bx_{t};\D) - \frac{\theta\gamma}{2}\|\nabla f(\bx_t;\D)\|_2^2 
\notag\\
& 
+
\frac{\theta\gamma-L(\theta\gamma)^2}{2}
\|\frac{1}{n}\sum\nolimits_{i=1}^{n}\nabla f(x_{i,t};\D_{i}) - \nabla f(\bx_t;\D)\|_2^2  \notag\\
&
+\frac{L(\theta\gamma)^2}{2} \Eb[\|\frac{1}{n}\sum\nolimits_{i=1}^{n}\nabla \tf(x_{i,t}|\Bc_{i,t})-\frac{1}{n}\sum\nolimits_{i=1}^{n}\nabla f(x_{i,t};\D)\|_2^2|\Fc_{t}]\notag\\
&
+
\frac{L(\theta\gamma)^2}{2}\|\frac{1}{n}\sum\nolimits_{i=1}^{n}\nabla f(x_{i,t};\D) - \nabla f(\bx_t;\D)\|_2^2 \notag\\
&
\le
f(\bx_{t};\D) - \frac{\theta\gamma}{2}\|\nabla f(\bx_t;\D)\|_2^2  
\notag\displaybreak[4]\\
&
+
\frac{\theta\gamma}{2}
\|\frac{1}{n}\sum\nolimits_{i=1}^{n}\nabla f(x_{i,t};\D_{i}) - \nabla f(\bx_t;\D)\|_2^2  \notag\\
&
~~
+\frac{L(\theta\gamma)^2}{2}[\big(\frac{1}{n\gamma}\big)^2(\frac{1}{p}-1)\|\L_\gamma(\x_t;\D)\|_2^2 + \frac{1}{np}(\frac{\tilde{\sigma}^2}{m\tau} + \sigma^2d)] \notag\\
&
\le
f(\bx_{t};\D) - \frac{\theta\gamma}{2}\|\nabla f(\bx_t;\D)\|_2^2  \notag\\
&
+
\frac{\theta\gamma}{2}
\|\frac{1}{n}\sum\nolimits_{i=1}^{n}\nabla f(x_{i,t};\D_{i}) - \nabla f(\bx_t;\D)\|_2^2  \notag\\
&
+\frac{L\theta^2}{2n^2}(\frac{1}{p}-1)\|\L_\gamma(\x_t;\D)\|_2^2 + \frac{L(\theta\gamma)^2}{2np}(\frac{\tilde{\sigma}^2}{m\tau} + \sigma^2d) \notag\\
&
\stackrel{(a)}{\le}
f(\bx_{t};\D) - \frac{\theta\gamma}{2}\|\nabla f(\bx_t;\D)\|_2^2  \notag\\
&
+
\frac{\theta\gamma}{2n}\sum\nolimits_{i=1}^{n}
\|\nabla f(x_{i,t};\D_{i}) - \nabla f(\bx_t;\D)\|_2^2  \notag\\
&
~~
+\frac{L\theta^2}{2n^2}(\frac{1}{p}-1)\|\L_\gamma(\x_t;\D)\|_2^2 + \frac{L(\theta\gamma)^2}{2np}(\frac{\tilde{\sigma}^2}{m\tau} + \sigma^2d) \notag\\
&
\stackrel{(b)}{\le}
f(\bx_{t};\D) - \frac{\theta\gamma}{2}\|\nabla f(\bx_t;\D)\|_2^2  
+
\frac{\theta\gamma L}{2n}
\|\x_t-\bxx_t\|_2^2  \notag\\
&
+\frac{L\theta^2}{2n^2}(\frac{1}{p}-1)\|\L_\gamma(\x_t;\D)\|_2^2 + \frac{L(\theta\gamma)^2}{2np}(\frac{\tilde{\sigma}^2}{m\tau} + \sigma^2d) 
\end{align} 
where (a) is by the Jensen's inequality, (b) is by the $L$-Lipschitz continuous gradient $\nabla f(\x;\zeta)$.

Taking full expectation and plugging the result in step 2, we have:
\begin{align}
&
\Eb[f(\bx_{t+1};\D) - f(\bx_{t};\D)]\notag\\
&
\stackrel{}{\le}
\Eb[
- \frac{\theta\gamma}{2}\|\nabla f(\bx_t;\D)\|_2^2
+ 
\frac{\theta^3\gamma L}{n}\sum_{s = 0 }^{t-1} \beta_\theta^{2(t-1-s)} (\frac{1}{p}-1)\|\nabla \L(\x_s;\D)\|^2  
 \notag\\
&
+
\frac{\theta\gamma LC_3}{n}\big(\frac{\gamma }{1-\beta}\big)^2
+ 
\frac{\theta^2\gamma^3 LC_2}{n(1-\beta)}(\frac{1}{p}-1)  
+
\frac{L\theta^2}{2n^2}(\frac{1}{p}-1)\|\L_\gamma(\x_t;\D)\|_2^2 \notag\\
&
+ \frac{L(\theta\gamma)^2}{2np}(\frac{\tilde{\sigma}^2}{m\tau} + \sigma^2d) ]
\end{align} 

Telescope the above inequality from $1$ to $T$:
\begin{align}
&
\sum_{t=0}^{T-1}\Eb[\frac{\theta\gamma}{2}\|\nabla f(\bx_t;\D)\|_2^2] \notag\\
&
\le
f(\bx_{0};\D) -  \Eb[f(\bx_{T};\D)] + 
\frac{\theta^3\gamma L}{n}\sum_{t=1}^{T-1}\sum_{s = 0 }^{t-1} \beta_\theta^{2(t-1-s)} (\frac{1}{p}-1)\times \notag\\
& 
\Eb[\|\nabla \L(\x_s;\D)\|^2]    
+
\frac{T\theta\gamma LC_3}{n}\big(\frac{\gamma }{1-\beta}\big)^2
+ 
\frac{T\theta^2\gamma^3 LC_2}{n(1-\beta)}(\frac{1}{p}-1) \notag\\
& 
+
\sum_{t=0}^{T-1}\frac{L\theta^2}{2n^2}(\frac{1}{p}-1)\Eb\|\L_\gamma(\x_t;\D)\|_2^2 
+ 
\frac{LT(\theta\gamma)^2}{2np}(\frac{\tilde{\sigma}^2}{m\tau} + \sigma^2d)\notag\\
&
\stackrel{(a)}{\le}
f(0;\D) -  f(x^*_{\D};\D)
+ 
\frac{\theta^3\gamma L}{n}\sum_{s = 0 }^{T-2}\sum_{t=s+1}^{T-1} \beta_\theta^{2(t-1-s)} (\frac{1}{p}-1)\times \notag\\
&
\Eb[\|\nabla \L(\x_s;\D)\|^2]
 +
\frac{T\theta\gamma LC_3}{n}\big(\frac{\gamma }{1-\beta}\big)^2  
+ 
\frac{T\theta^2\gamma^3 LC_2}{n(1-\beta)}(\frac{1}{p}-1) \notag\displaybreak[4]\\
&
+
\sum_{t=0}^{T-1}\frac{L\theta^2}{2n^2}(\frac{1}{p}-1)\Eb\|\L_\gamma(\x_t;\D)\|_2^2 
+ 
\frac{LT(\theta\gamma)^2}{2np}(\frac{\tilde{\sigma}^2}{m\tau} + \sigma^2d)\notag\\
&
\stackrel{}{\le}
f(0;\D) -  f(x^*_{\D};\D)
+ 
\frac{\theta^3\gamma L}{n(1-\beta_\theta^2)}\sum_{s = 0 }^{T-2} (\frac{1}{p}-1)\|\nabla \L(\x_s;\D)\|^2\notag\\
&
 +
\frac{T\theta\gamma LC_3}{n}\big(\frac{\gamma }{1-\beta}\big)^2  
+ 
\frac{T\theta^2\gamma^3 LC_2}{n(1-\beta)}(\frac{1}{p}-1) 
+
\sum_{t=0}^{T-1}\frac{L\theta^2}{2n^2}(\frac{1}{p}-1)\times \notag\\
&
\Eb\|\L_\gamma(\x_t;\D)\|_2^2 
+ 
\frac{LT(\theta\gamma)^2}{2np}(\frac{\tilde{\sigma}^2}{m\tau} + \sigma^2d)\notag\displaybreak[3]\\
&
\stackrel{(b)}{\le}
f(0;\D) -  f(x^*_{\D};\D)
+ 
\frac{\theta^2\gamma L}{n(1-\beta)}\sum_{t = 0 }^{T-1} (\frac{1}{p}-1)\Eb[\|\nabla \L(\x_t;\D)\|^2] \notag\\
&
 +
\frac{T\theta\gamma LC_3}{n}\big(\frac{\gamma }{1-\beta}\big)^2  
+ 
\frac{T\theta^2\gamma^3 LC_2}{n(1-\beta)}(\frac{1}{p}-1) \notag\\
&
+
\sum_{t=0}^{T-1}\frac{L\theta^2}{2n^2}(\frac{1}{p}-1)\Eb\|\L_\gamma(\x_t;\D)\|_2^2 
+ 
\frac{LT(\theta\gamma)^2}{2np}(\frac{\tilde{\sigma}^2}{m\tau} + \sigma^2d)\notag\\
&
\stackrel{}{\le}
f(0;\D) -  f(x^*_{\D};\D)
+ 
(\frac{\theta^2\gamma L}{n(1-\beta)} + \frac{L\theta^2}{2n^2})(\frac{1}{p}-1) \times \notag\\
&
\sum_{t = 0 }^{T-1}\Eb[\|\nabla \L(\x_t;\D)\|^2]
 +
\frac{T\theta\gamma LC_3}{n}\big(\frac{\gamma }{1-\beta}\big)^2  \notag\\
&
+ 
\frac{T\theta^2\gamma^3 LC_2}{n(1-\beta)}(\frac{1}{p}-1)
+ 
\frac{LT(\theta\gamma)^2}{2np}(\frac{\tilde{\sigma}^2}{m\tau} + \sigma^2d)\notag\\
&
\stackrel{(c)}{\le}
f(0;\D) -  f(x^*_{\D};\D)
+
\frac{T\theta\gamma LC_3}{n}\big(\frac{\gamma }{1-\beta}\big)^2 
+ 
\frac{T\theta^2\gamma^3 LC_2}{n(1-\beta)}(\frac{1}{p}-1)\notag\\
&
+ 
\frac{LT(\theta\gamma)^2}{2np}(\frac{\tilde{\sigma}^2}{m\tau} + \sigma^2d)
+ 
(\frac{\theta^2\gamma L}{n(1-\beta)} + \frac{L\theta^2}{2n^2})(\frac{1}{p}-1)\times \notag\\
&
\Big(\frac{2pn\gamma C_1}{\big(2p\theta-(1-\lambda_n+\gamma L)\theta^2\big)} +
\frac{(1-\lambda_n+\gamma L)\theta\gamma^2 TC_2}{2p-(1-\lambda_n+\gamma L)\theta}\Big)
\end{align} 
where (a) is by $\Eb[f(\bx_T);\D] \ge f(x^*_{\D};\D)$ and Fubini's theorem, (b) is by $\beta_\theta \in (0,1)$ and Lemma \ref{Lemma: beta/theta} and (c) is from step 1. Hence, 
\begin{align}
&\sum_{t=0}^{T-1}\|\nabla f(\bx_t;\D)\|_2^2
\stackrel{}{\le}
\frac{2C_1}{\theta\gamma} 
+
\frac{2T LC_3}{n}\big(\frac{\gamma }{1-\beta}\big)^2 
+ 
\frac{2T\theta\gamma^2 LC_2}{n(1-\beta)}(\frac{1}{p}-1)\notag\\
&
+ 
\frac{LT\theta\gamma}{n^2p}C_2
+ 
(\frac{2\theta\gamma L}{n(1-\beta)} + \frac{L\theta}{n^2})(\frac{1}{p}-1)\Big(\frac{2pn C_1}{\big(2p\theta-(1-\lambda_n+\gamma L)\theta^2\big)} \notag\\
&
+
\frac{(1-\lambda_n+\gamma L)\theta\gamma TC_2}{2p-(1-\lambda_n+\gamma L)\theta}\Big)
\end{align} 
where $C_1 = f(0;\D) -  f(x^*_{\D};\D),$ $C_2 = n\tilde{\sigma}^2/m\tau+nd\sigma^2$ and $C_3 = (nG)^2+(nd\sigma)^2$ are constants.
\end{proof}

\subsection{Proof of Corollary \ref{Cor: Convergence}}
\begin{proof}
By setting $\theta = \min\{p/(1-\lambda_n+\gamma L),p/2\}\le 1,$ we have 
\begin{align}
&(1-\lambda_n+\gamma L)\theta = \min\{p,p(1-\lambda_n+\gamma L)/2\} \le p \\
&1/\theta =\max\{(1-\lambda_n+\gamma L)/p,2/p\}\stackrel{(a)}{\le} 3/p 
\end{align}
where (a) is by $1-\lambda_n+\gamma L \le 3$ with a very small $\gamma$ (i.e. large enough $T$).
Thus for (I) in (\ref{Eq: convergence error}), $2C_1/\theta\gamma T \le 6C_1/\gamma Tp = O(1/\gamma Tp);$ (II) is $O\big(n\gamma^2/(1-\beta)^2\big);$ for (III) is 
\begin{align}
\frac{2\theta\gamma^2 LC_2}{n(1-\beta)}(\frac{1}{p}-1)
+ 
\frac{L\theta\gamma C_2}{n^2p}&
\!\le\!
\frac{\gamma^2 LC_2}{n(1-\beta)}
+ 
\frac{L\gamma C_2}{2n^2} 
\!=\! 
O\big(\frac{\gamma^2}{(1-\beta)} + \frac{\gamma}{n} \big);
\end{align}
and (IV) is 
{\small{\begin{align}
&
(\frac{2\gamma L}{n(1 \!-\!\beta)}  \!+\! \frac{L}{n^2})(\frac{1}{p} \!-\!1)\Big(\frac{2pn C_1}{\big(2p-(1-\lambda_n+\gamma L)\theta\big)T} \!+\!
\frac{(1\!-\!\lambda_n\!+\!\gamma L)\theta^2\gamma C_2}{2p\!-\!(1\!-\!\lambda_n\!+\!\gamma L)\theta}\Big) \notag\\
&
\le
(\frac{2\gamma L}{n(1-\beta)} + \frac{L}{n^2})\Big(\frac{2n C_1}{Tp} +
2\gamma C_2\Big)
=
O\big((\frac{\gamma}{1-\beta} + \frac{1}{n})(\frac{1}{Tp} +
2\gamma )\big)
\end{align}}}
To summarize, the convergence error has the order:
\begin{align}
&O\Big(\frac{1}{\gamma Tp}+\frac{n\gamma^2}{(1-\beta)^2}+\frac{\gamma^2}{(1-\beta)} + \frac{\gamma}{n}+(\frac{\gamma}{1-\beta} + \frac{1}{n})(\frac{1}{Tp} +
2\gamma)
\Big)\notag\\
&
= O\Big(\frac{1}{\gamma Tp}+\frac{n\gamma^2}{(1-\beta)^2}+\frac{\gamma^2}{(1-\beta)} + \frac{\gamma}{n}+\frac{\gamma}{(1-\beta)Tp} + \frac{1}{nTp}\Big)\notag\\
&
= O\Big(\frac{1}{\gamma Tp}+\frac{n\gamma^2}{(1-\beta)^2}+ \frac{\gamma}{n}+\frac{\gamma}{(1-\beta)Tp} + \frac{1}{nTp}\Big)
\end{align}
Set $\gamma = c\sqrt{n\log(T)/T},$  then order of the convergence error is:
{\small{\begin{align}
O\Big(\frac{1}{\sqrt{n\log(T)Tp}}+\frac{n^2\log(T)}{(1-\beta)^2Tp}+ \sqrt{\frac{\log(T)}{nT}}+\frac{\sqrt{n\log(T)}}{(1-\beta)\sqrt{(Tp)^3}} + \frac{1}{nTp}\Big).
\end{align}}}
With the large iteration number, i.e. $T/\log(T)^4 > n^5/(1-\beta)^4,$ then the order of the convergence error is bounded by:
\begin{align}
O\Big(\frac{1}{ \sqrt{nT}}+ \sqrt{\frac{\log(T)}{nT}}+ \frac{\log(T)}{nT}\Big) = O\Big(\sqrt{\frac{\log(T)}{nT}}\Big).
\end{align}

\end{proof}

\end{document}